\documentclass[11pt,reqno]{amsart}
\usepackage[foot]{amsaddr}

\usepackage[left=2.5cm, right=2.5cm, top=2cm]{geometry}

\usepackage[english]{babel}
\usepackage[colorinlistoftodos]{todonotes}
\usepackage{amsmath,amssymb,amsthm,amsfonts,mathrsfs,amsxtra,bbm,comment}

\usepackage{tikz-cd}

\usepackage{tikz}
\usetikzlibrary{arrows,positioning}
\usetikzlibrary{decorations.pathmorphing}
\usetikzlibrary{decorations.markings}
\usetikzlibrary{patterns}
\usetikzlibrary{automata}
\usepackage{tikz-cd}

\tikzset{->-/.style={decoration={
			markings,
			mark=at position #1 with {\arrow{latex}}},postaction={decorate}}}
	
\tikzset{-<-/.style={decoration={
			markings,
			mark=at position #1 with {\arrowreversed{latex}}},postaction={decorate}}}

\usetikzlibrary{shapes.misc}\tikzset{cross/.style={cross out, draw, 
         minimum size=2*(#1-\pgflinewidth), 
         inner sep=0pt, outer sep=0pt}}

\usepackage{pgfplots}
\pgfplotsset{compat = 1.11}

\usepackage[hyperfootnotes=true,
        pdffitwindow=true,
        plainpages=false,
        pdfpagelabels=true,
        pdfpagemode=UseOutlines,
        pdfpagelayout=SinglePage,
        hyperindex,]{hyperref}

\numberwithin{equation}{section}

\theoremstyle{plain}
	\newtheorem{theorem}{Theorem}

	\newtheorem{proposition}[theorem]{Proposition}
\theoremstyle{definition}
	
	\newtheorem{example}{Example}
        
\theoremstyle{remark}
	\newtheorem{remark}{Remark}

\newcommand{\R}{\mathbb R}

\newcommand{\dv}{\mathrm{d}}

\newcommand{\f}{\mathcal F}

\newcommand{\sgn}{\mathrm{sgn}}

\newcommand{\ind}{\mathbbm{1}}

\newcommand{\re}{\mathrm{Re}\,}

\renewcommand{\Re}{\mathrm{Re}\,}
\renewcommand{\Im}{\mathrm{Im}\,}

\title{Biorthogonal ensembles of derivative type}

\author{Tom Claeys}
\email[TC]{[TC]tom.claeys@uclouvain.be}
\address[TC]{Institut de Recherche en Mathématique et Physique, UCLouvain, Chemin du Cyclotron 2, 1348 Louvain-la-Neuve, Belgium}

\author{Jiyuan Zhang}
\email[JZ]{[JZ]jiyuanzhang@scut.edu.cn}
\address[JZ]{School of Mathematics, South China University of Technology, China}

\date{\today}

\begin{document}

\begin{abstract}
In this paper, we prove that biorthogonal ensembles on the real line with a specific derivative structure admit an explicit correlation kernel of double contour integral form. We will demonstrate that this expression is a valuable starting point for asymptotic analysis and that our class of biorthogonal ensembles admits a large variety of limit kernels, by proving that two new classes of limit kernels can occur. The first type is a deformation of the hard edge Bessel kernel which arises in polynomial ensembles describing the eigenvalues of the sum of two random matrices, while the second type arises for Muttalib-Borodin type deformations of polynomial ensembles.
\end{abstract}

\maketitle

\section{Introduction}
 
{\em Biorthogonal ensembles}
\cite{Borodin} are symmetric probability distributions on $\mathbb R^N$, with $N$ a positive integer, which take the form
\begin{equation}\label{BE}
    \frac{1}{Z_N}\det[f_j(x_k)]_{j,k=1}^N\det[g_k(x_j)]_{j,k=1}^N \mathrm d x_1\ldots \mathrm dx_N,\quad x_1,\ldots,x_N\in\R,
\end{equation}
for certain families of functions $f_1,\ldots, f_N$ and $g_1,\ldots, g_N$. 
Such probability distributions are omnipresent in random matrix theory as eigenvalue distributions, but also arise in the study of random partitions, random tilings, random growth phenomena, polymer models and 
last passage percolation.

Biorthogonal ensembles are special cases of $N$-point determinantal point processes. 
Indeed, the $n$-point correlation functions of a biorthogonal ensemble, defined by
\begin{equation}\label{def:correlation}\rho_n(x_1,\ldots, x_n):=\frac{N!}{(N - n)!}\frac{1}{Z_N}\int_{\mathbb R^{N - n}}  
\det[f_j(x_k)]_{j,k=1}^N\det[g_k(x_j)]_{j,k=1}^N
\prod_{k = n +1}^N \mathrm d x_{k},\quad n=1,\ldots, N,
\end{equation}
can be written as a determinant involving a kernel function $K_N$,
\begin{equation}\label{correlationfunction}
\rho_n(x_1,\ldots, x_n)=\det\Big[K_N(x_m,x_k)\Big]_{m,k=1}^n.
\end{equation}
Here, the correlation kernel $K_N$ can be taken of the form 
\begin{equation}\label{kernel}
K_N(x,x')=\sum_{k=1}^N\phi_k(x')\psi_k(x),
\end{equation}
where $\phi_1,\ldots, \phi_N$ have the same linear span as $f_1,\ldots, f_N$, and $\psi_1,\ldots, \psi_N$ have the same linear span as $g_1,\ldots, g_N$, and they satisfy the biorthogonality relations
\begin{equation}\label{biorth}
\int_{\mathbb R}\phi_k(x)\psi_m(x)\mathrm dx=\delta_{km}.
\end{equation}
This explains why such ensembles are called biorthogonal.

Important subclasses of biorthogonal ensembles are {\em polynomial ensembles} \cite{KuijlaarsStivigny}, which correspond to $f_j(x)=x^{j-1}$ with $g_k$ arbitrary, and {\em orthogonal polynomial ensembles}, which correspond to $f_j(x)=x^{j-1}$ and $g_k(x)=x^{k-1}w(x)$ for some function $w$.
Orthogonal polynomial ensembles describe eigenvalue distributions of unitary invariant Hermitian random matrices \cite{Deift}, while examples of polynomial ensembles occur as eigenvalue distributions for certain types of sums and products of random matrices.
In these random matrix models as well as in most of the other statistical physics models connected to biorthogonal ensembles, the main question of interest is to understand the large $N$ asymptotics.
This question can be addressed by a large $N$ asymptotic analysis of the correlation kernel $K_N$, which encodes, through the correlation functions, relevant statistics of the random point configurations $x_1,\ldots, x_N$.

However, so far, large $N$ asymptotic analysis of the correlation kernel has only been successful in special types of biorthogonal ensembles \eqref{BE}. 
A first type consists of biorthogonal ensembles in which exact expressions for the correlation kernel, usually in double contour integral form, are available. Such ensembles arise as eigenvalue distributions in many classical random matrix ensembles such as the Gaussian Unitary Ensemble and Laguerre Unitary Ensemble with external source \cite{BBP, BleherKuijlaars, BrezinHikami, BrezinHikami2, DesrosiersForrester, DiekerWarren, ElKaroui, Johansson,LiuWang}, as well as in eigenvalue distributions of certain sums and products of independent classical random matrices \cite{Forrester2, ForresterLiu, KieburgKuijlaarsStivigny, KuijlaarsZhang, ZyczkowskiSommers,ForresterIpsenLiu,LiuWangWang,LiuA,LiuB}.
In such cases, saddle point methods enable one to study large $N$ asymptotics of the correlation kernel, and to obtain universal limit kernels like the sine, Airy and Bessel kernels, and also more complicated kernels like the Pearcey kernel and Meijer-G function kernels.
A second type consists of biorthogonal ensembles whose correlation kernel does not admit an explicit expression, but can instead be characterized by a Riemann-Hilbert problem which can be analyzed asymptotically. This is the case in orthogonal polynomial ensembles \cite{Deift, DKMVZ1, DKMVZ2}, for which possible limit kernels are now well understood, and also in some multiple orthogonal polynomial ensembles \cite{KuijlaarsMOP,Wang,WangZhang}.

In view of the efforts of many researchers in the past years to extend the range of biorthogonal ensembles which we can analyze asymptotically, it is natural to ask under which conditions a biorthogonal ensemble admits a correlation kernel which can be written in double contour integral form.
We will show that this is the case if (and only if, in some sense) the biorthogonal ensemble possesses a specific derivative structure.
To explain this derivative structure, let us first recall the notion of polynomial ensembles of derivative type, also known as P\'olya ensembles \cite{ForsterKieburgKosters, kieburg17, KK16, KieburgZhang}.

{\em Polynomial ensembles of additive derivative type} are biorthogonal ensembles of the form
\begin{equation}\label{PE}
    \frac{1}{Z_N}\det[x_k^{j-1}]_{j,k=1}^N\det[(-\partial_{x_j})^{k-1}w(x_j)]_{j,k=1}^N \mathrm d x_1\ldots \mathrm d x_N,\quad x_1,\ldots,x_N\in\R,\ Z_N>0,
\end{equation}
for some function $w:\mathbb R\mapsto [0,+\infty)$ which is at least $N-1$ times differentiable in $L^1$-sense. 
Simple examples of such probability measures arise as eigenvalue distributions in classical random matrix models. 
\begin{example}
Set $w(x)=x^{N-1+\nu}e^{-x}\ind_{(0,+\infty)}(x)$ with $\nu\geq 0$. Then
we have \begin{equation*}\partial_{x}^{k-1}w(x)=p_{k-1}(1/x)x^{N-1+\nu}e^{-x}\ind_{(0,+\infty)}(x)\end{equation*} for some polynomial $p_{k-1}$ of degree $k-1$. By performing elementary row operations, we obtain that
\begin{equation*}\det[(-\partial_{x_j})^{k-1}w(x_j)]_{j,k=1}^N=c_N\det[x_j^{k-1}]_{j,k=1}^N \ \prod_{j=1}^N x_j^\nu e^{-x_j}\ind_{(0,+\infty)}(x_j),\end{equation*}
for some constant $c_N$. Recognizing the Vandermonde determinant, we can now rewrite the ensemble \eqref{PE} as 
\begin{equation}\label{PE-LUE}
    \frac{1}{\hat Z_N}\Delta(x)^2\prod_{j=1}^N x_j^\nu e^{-x_j}\mathrm d x_j
    ,\quad x_1,\ldots,x_N>0,\qquad \Delta(x)=\prod_{1\leq j<k\leq N}(x_k-x_j),
\end{equation}
which is (see, e.g., \cite[Section 3.2]{Forrester_book}) the eigenvalue distribution of the $N\times N$ Laguerre Unitary Ensemble (LUE), or Wishart-Laguerre Ensemble, of Hermitian positive-definite random matrices. Notice that for $\nu=0$, $w$ is precisely $N-1$ times differentiable in $L^1$-sense, with the last derivative being integrable but not continuous.
\end{example}

\begin{example} If we set $w(x)=e^{-x^2/2}$, then $\partial_{x}^{k-1}w(x)=p_{k-1}(x)e^{-x^2/2}$ for a polynomial $p_{k-1}$ of degree $k-1$. We then have
\begin{equation*}\det[(-\partial_{x_j})^{k-1}w(x_j)]_{j,k=1}^N=c_N\det[x_j^{k-1}]_{j,k=1}^N \ \prod_{j=1}^N  e^{-x_j^2/2},\end{equation*}
for some $c_N$,
such that \eqref{PE} is 
given by
\begin{equation}\label{PE-GUE}
    \frac{1}{\hat Z_N}\Delta(x)^2\prod_{j=1}^N e^{-x_j^2/2}\mathrm d x_j
    ,\quad x_1,\ldots,x_N\in\R.
\end{equation}
This is (see, e.g., \cite{Forrester_book, Mehta})
the eigenvalue distribution of the $N\times N$ Gaussian Unitary Ensemble (GUE).
\end{example}
While these are the two most classical examples of polynomial ensembles of additive derivative type, others occur as eigenvalue distributions of sums of independent random matrices \cite{kieburg17, ClaeysKuijlaarsWang, KuijlaarsPol} and in dynamical random matrix ensembles \cite{Assiotis, Assiotis2}. The class of functions $w$ for which \eqref{PE} defines a positive measure for every $N\in\mathbb N$ is rather restricted: it contains only those functions $w$ for which 
\begin{equation}
\frac{\det[(-\partial_{x_j})^{k-1}w(x_j)]_{j,k=1}^N}{\Delta(x)}\geq 0,\quad\mbox{for every $N\in\mathbb N$.}
\end{equation}
Such functions are {\em P\'olya frequency functions} of infinite order \cite{karlin,kieburg17} and for this reason, ensembles of the form \eqref{PE} are also called {\em P\'olya ensembles}.

{\em Polynomial ensembles of multiplicative derivative type} are biorthogonal ensembles on $(0,\infty)$ of the form
\begin{equation}\label{PEmult}
    \frac{1}{Z_N}\Delta(y)\det[(-y_j\partial_{y_j})^{k-1}\tilde w(y_j)]_{j,k=1}^N \mathrm d y_1\ldots \mathrm d y_N,\quad y_1,\ldots,y_N>0,
\end{equation}
for some weight function $\tilde w$.
There are again classical random matrix models which give rise to such ensembles; more complicated examples
 arise as eigenvalue distributions for products of random matrices \cite{AkemannIpsenKieburg, AkemannKieburgWei, KK16, KieburgKuijlaarsStivigny, KuijlaarsZhang}.

\begin{example}\label{example:LUEmult}
For $\tilde w(y)=y^\nu e^{-y}$ with $\nu>-1$, \eqref{PEmult} is the joint eigenvalue distribution \eqref{PE-LUE} of the LUE. Curiously, the LUE is thus a polynomial ensemble that has both an additive derivative structure and a multiplicative derivative structure.
\end{example}
\begin{example}\label{example:JUE}
For $\tilde w(y)={y^{\mu}(1-y)^{\nu}\ind_{(0,1)}(y)}$ with $\mu>-1, \nu>N-1$, \eqref{PEmult} is the eigenvalue distribution in the Jacobi Unitary Ensemble (JUE) \cite[Section 3.6]{Forrester_book}.
\end{example}
\begin{example}\label{example:CLUE}
If we set $\tilde w(y)=y^{\beta}(1+y)^{-\beta-\gamma-1}$ with $\beta>-1, \gamma>N-1$, \eqref{PEmult} is the eigenvalue distribution in the Cauchy-Lorentz Unitary Ensemble (C-LUE), see \cite[Section 2.5]{Forrester_book} and \cite{KieburgetalCL}.
\end{example}

We introduce a deformation of \eqref{PE}, in which the Vandermonde determinant $\det[x_k^{j-1}]_{j,k=1}^N$ is modified, of the form
\begin{equation}\label{BEDT}
    \frac{1}{Z_N(a_1,\ldots, a_N)}\frac{\det[e^{a_jx_k}]_{j,k=1}^N}{\Delta(a)}\det[(-\partial_{x_j})^{k-1}w(x_j)]_{j,k=1}^N \mathrm d x_1\ldots \mathrm d x_N,\quad x_1,\ldots,x_N\in\R,
\end{equation}
with 
\begin{equation}\label{def:ZN}
Z_N(a_1,\ldots, a_N)=\int_{\mathbb R^N}\frac{\det[e^{a_jx_k}]_{j,k=1}^N}{\Delta(a)}\det[(-\partial_{x_j})^{k-1}w(x_j)]_{j,k=1}^N \mathrm d x_1\ldots \mathrm dx_N,\end{equation}
for some function $w:\mathbb R\mapsto [0,+\infty)$ which is at least $N-1$ times differentiable in $L^1$-sense, and for distinct real values of $a_1,\ldots, a_N$. If they are not distinct, this distribution is not well-defined because both determinants $\Delta(a)$ and 
$\det[e^{a_jx_k}]_{j,k=1}^N$
vanish. However, 
we observe that
$\frac{\det[e^{a_jx_k}]_{j,k=1}^N}{\Delta(x)\Delta(a)}$ is proportional to the Harish-Chandra-Itzykson-Zuber integral \cite{HC,IZ}, such that it is well-defined and positive for all real $a_1,\ldots, a_N$ and $x_1,\ldots, x_N$.
Thus, the ratio 
$\frac{\det[e^{a_jx_k}]_{j,k=1}^N}{\Delta(a)}$ remains well-defined when taking confluent limits where, for instance, $a_{j_1},\ldots a_{j_k}\to a$, $k\in\{2,\ldots, N\}$, and this results in well-defined ensembles.
Moreover,
\eqref{BEDT} is positive if and only if \eqref{PE} is positive, such that the connection between positivity of the ensemble and P\'olya frequency functions is preserved.
We will call ensembles of the form \eqref{BEDT}, as well as their degenerate confluent limits, {\em biorthogonal ensembles of derivative type}.
There are two notable special cases.
\begin{itemize}\item The fully confluent case of \eqref{BEDT} where $a_1=\cdots=a_N= 0$ is equal to the polynomial ensemble of additive derivative type \eqref{PE}. We will show this in Section \ref{sec:confluent}. 
\smallskip
\item If we set $a_j=j$, change variables $y_j=e^{x_j}$ and set $\tilde w(y)=w(\log y)$, then \eqref{BEDT} is the polynomial ensemble of multiplicative derivative type \eqref{PEmult}.
\end{itemize}
Our class of biorthogonal ensembles of derivative type thus comprises both polynomial ensembles of additive and multiplicative derivative types as special cases.
Biorthogonal ensembles of derivative type which are not polynomial ensembles of derivative type occur in random matrix models with external source.
\begin{example}\label{example:GUE+}
If we set $w(x)=e^{-x^2/2}$, the biorthogonal ensemble of derivative type \eqref{BEDT} is the joint eigenvalue distribution of the eigenvalues of a random matrix in the GUE with external source \cite{BrezinHikami, Johansson}. This is a polynomial ensemble, but it is in general not of derivative type. A random matrix in this ensemble can be written as $M+A$, where $A$ is an $N\times N$ Hermitian matrix with eigenvalues $a_1,\ldots, a_N$, and $M$ is an $N\times N$ GUE matrix.
An associated biorthogonal system can be constructed in terms of multiple Hermite polynomials \cite{BleherKuijlaars}. 
\end{example}
\begin{example}\label{example:LUE+}
If we set $w(x)=x^{N-1+\nu}e^{-x}\ind_{(0,+\infty)}(x)$, \eqref{BEDT} is the eigenvalue distributions of eigenvalues of a random matrix in the LUE with external source $a_1,\ldots, a_N$, also known as Wishart ensemble or ensemble of complex Gaussian sample covariance matrices \cite{BBP, ElKaroui}. Again, this is a polynomial ensemble but in general not of derivative type. Here, a biorthogonal system can be constructed using multiple Laguerre polynomials \cite{BleherKuijlaars}.
\end{example}

General biorthogonal ensembles of derivative type do not arise directly as eigenvalue distributions of random matrices, although there exist methods to construct random matrices with prescribed biorthogonal ensembles as eigenvalue distributions, see e.g.\ \cite{Cheliotis, ForresterWang}.

We will prove that biorthogonal ensembles of derivative type in general admit an {\bf explicit double contour integral representation for their correlation kernel} (see Theorem \ref{thm:algebraic} below). This provides a powerful starting point for large $N$ asymptotic analysis of biorthogonal ensembles of derivative type. Moreover, the class of biorthogonal ensembles of derivative type is rich in view of possible limit kernels that can occur when taking scaling limits. Indeed, we will show that {\bf two new classes of limit kernels} occur, the first one for additive perturbations of the LUE (see Theorem \ref{theorem:LUE+} below), the second one for Mutallib-Borodin deformations of polynomial ensembles of derivative type (see Theorem \ref{thm:MB} below). 

\section{Statement of results}\label{section:results}

\subsection{Double contour integral expression for the correlation kernel}\label{section:resultsalg}
Our first result states that if $w$ is sufficiently smooth and fast decaying at infinity, the associated biorthogonal ensemble \eqref{BEDT} of derivative type admits a correlation kernel of the form
\begin{equation}\label{kernel0}
	K_N(x,x'):=\int_{\Sigma_N}\frac{\dv u}{2\pi i}\int_{c+i\mathbb R}\frac{\dv v}{2\pi i}\frac{W(v)\prod_{j=1}^N(v-a_j)}{W(u)\prod_{j=1}^N(u-a_j)}\frac{e^{-xv+x'u}}{v-u}.
\end{equation}
Here $W$ is an analytic function in a vertical strip $c_-<\Re z<c_+$ of the complex plane containing the imaginary line, $c_-<c<c_+$, and $\Sigma_N$ is a closed positively oriented curve in the strip $c_-<\Re z<c_+$, not intersecting $c+i\mathbb R$, and encircling $a_1,\ldots,a_N$ but not enclosing zeros of $W$, see Figure \ref{figure: contours}.

Conversely, if $W(z)$ is analytic in a sufficiently wide vertical strip in the complex plane satisfying sufficient decay conditions at infinity, then \eqref{kernel0} is the kernel of a biorthogonal ensemble of derivative type, but we should notice that this biorthogonal ensemble does not need to be positive, in general it will be a complex-valued biorthogonal measure.

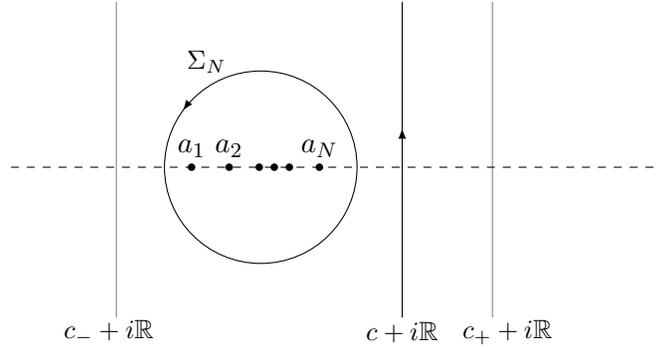
\begin{figure}[t]
\begin{center}
	\begin{tikzpicture}
		\fill (-2,0) circle (0.05cm);
		\node at (-2,0.25) {$a_1$};
		\fill (-1.5,0) circle (0.05cm);
		\node at (-1.5,0.25) {$a_2$};
		\fill (-.9,0) circle (0.05cm);
        \fill (-1.1,0) circle (0.05cm);
        \fill (-.7,0) circle (0.05cm);
        \fill (-0.3,0) circle (0.05cm);
		\node at (-0.3,0.25) {$a_N$};
        \draw[gray] (-3,-2) to [out=90, in=-90] (-3,2.2);
        \draw[gray] (2,-2) to [out=90, in=-90] (2,2.2);
        \draw[black,->-=0.6] (0.8,-2) to [out=90, in=-90] (0.8,2.2);
		\draw[dashed,-=1,black] (-4.4,0) to [out=0, in=-180] (4.4,0);
		\node at (-1.8,1.4) {\small $\Sigma_N$};
		\node at (0.8,-2.2) {\small $c+i\mathbb R$};	
        \node at (-3.1,-2.2) {\small $c_-+i\mathbb R$};	
        \node at (2.2,-2.2) {\small $c_++i\mathbb R$};	
        \draw[->-=0.4,black] (0.2,0) arc (0:360:1.28cm);
\end{tikzpicture}
    \caption{The vertical strip $c_-<\Re z<c_+$, the points $a_1,\ldots, a_N$, and the integration contours $\Sigma_N$ and $c+i\mathbb R$ in \eqref{kernel0}.}
    \label{figure: contours}
\end{center}
\end{figure}

The identification between $w$ (in~\eqref{BEDT}) and $W$ goes via the Fourier transform \begin{equation}
\label{def:Fourier}
\mathcal F[w](t)=\int_{\mathbb R}e^{ixt}w(x)\mathrm dx \quad\mbox{with inverse}\quad \mathcal F^{-1}[f](x)=\frac{1}{2\pi}\int_{\mathbb R}e^{-ixt}f(t)\mathrm dt.
\end{equation}
We have
\begin{equation}\label{def:Ww}
W(z)=\mathcal F[w](-iz)=\int_{\mathbb R}e^{xz}w(x)\mathrm dx,\qquad w(x)=\frac{1}{2\pi}\int_{\mathbb R}e^{-ixt}W(it)\mathrm dt=\frac{1}{2\pi i}\int_{i\mathbb R}e^{xz}W(-z)\mathrm dz.
\end{equation}

Given $c_-<0<c_+$, let us define $\mathcal U_{c_-,c_+}^N$ as the space of all functions $w:\mathbb R\to\mathbb C$ satisfying the following conditions:
    \begin{enumerate}
        \item (regularity) $w$ belongs to the Sobolev space $W^{N-1,1}(\mathbb R)$ of $N-1$ times differentiable functions in $L^1$-sense: $w,w',\ldots, w^{(N-1)}\in L^1(\mathbb R)$;
        \smallskip
        \item (exponential decay) for every $c\in(c_-,c_+)$, there exists $\tilde c>0$ such that
        \begin{align}
            &|w^{(j)}(x)|\leq \tilde c e^{-c x},\qquad x\in\mathbb R,
        \end{align}
        for $j=0,1,\ldots,N-1$.
    \end{enumerate}

As we will show, the first relation in \eqref{def:Ww} maps the function space  $\mathcal U_{c_-,c_+}^N$ into the space of functions $\mathcal V_{c_-,c_+}^N$ containing all complex-valued functions $W$ for which:
\begin{enumerate}
\item (analyticity) $W$ is analytic in the strip $\mathcal S_{c_-,c_+}:=\{z\in\mathbb C:c_-<\Re z<c_+\}$;
\smallskip
\item (integrability) for any $\epsilon>0$, we have for $p=2$ and for $p=\infty$ that 
\begin{equation*}\sup_{c_-+\epsilon<c<c_+-\epsilon}\|z^{N-1}W(z)\|_{L^p(c+i\mathbb R)}<\infty;\end{equation*}
\item (decay) and moreover, for any $c_-<c<c_+$, 
\begin{equation*}\lim_{t\to \pm\infty}(c+it)^{N-1}W(c+it)=0.\end{equation*}
\end{enumerate}

\begin{theorem}
\label{thm:algebraic}
Let $c_-<0<c_+$ and $N\in\mathbb N$.
\begin{enumerate}\item For any $w\in\mathcal U_{c_-,c_+}^N$ and for any distinct $a_1,\ldots, a_N\in\mathcal S_{c_-,c_+}$, the complex-valued biorthogonal ensemble \eqref{BEDT} is well-defined, and it admits the correlation kernel \eqref{kernel0} with $W\in \mathcal V_{c_-,c_+}^N$ given by \eqref{def:Ww}, provided that $W(a_j)\neq 0$ for all $j=1,\ldots, N$. 
Moreover, 
    \begin{equation}\label{ZN0}	Z_N(a_1,\ldots,a_N)=N!\prod_{j=1}^NW(a_j).
    \end{equation}
\item Let $W\in\mathcal V_{c_-,c_+}^{N}$ and let $W$ be moreover such that
\begin{equation*}\sup_{c_-+\epsilon<c<c_+-\epsilon}\|z^{N-1}W(z)\|_{L^1(c+i\mathbb R)}<\infty,\end{equation*}
for any $\epsilon>0$.
For any $a_1,\ldots, a_N\in\mathcal S_{c_-,c_+}$ for which $W(a_j)\neq 0$ for all $j=1,\ldots, N$, the kernel $K_N$ given by \eqref{kernel0} is well-defined; if $a_1,\ldots, a_N$ are distinct, the complex-valued measure $\frac{1}{N!}\rho_N(x_1,\ldots, x_N)dx_1\ldots dx_N$ defined by \eqref{correlationfunction} is equal to \eqref{BEDT} with $w\in \mathcal U_{c_-,c_+}^N$ given by \eqref{def:Ww}; if $a_1=\ldots= a_N=0$,  $\frac{1}{N!}\rho_N(x_1,\ldots, x_N)dx_1\ldots dx_N$ is equal to \eqref{PE} with $w\in \mathcal U_{c_-,c_+}^N$ given by \eqref{def:Ww}.  
\end{enumerate}
\end{theorem}
\begin{remark}
Just like for positive biorthogonal ensembles, when we say that $K_N$ is a correlation kenel of a complex-valued biorthogonal ensemble, we mean that the correlation functions (or marginal densities) defined by \eqref{def:correlation}  are equal to \eqref{correlationfunction}.
Signed biorthogonal ensembles occur for instance in polymer models, see \cite{CC24}.
\end{remark}
\begin{remark}\label{remark:contourdef}
Let $a_1,\ldots, a_N\in \mathcal S_{c_-,c_+}$ and let $\Sigma_N$ enclose $a_1,\ldots, a_N$ without enclosing zeros of $W$. Then, for $W\in\mathcal V_{c_-,c_+}$, it is straightforward to check that \eqref{kernel0} is independent of $c$, as long as $c+i\mathbb R$ does not intersect $\Sigma_N$. For any $c\in(c_-,c_+)$ such that $c+i\mathbb R$ lies to the right of $\Sigma_N$, this independence follows from the analyticity and decay of $W$ by a contour deformation argument. If $c+i\mathbb R$ intersects $\Sigma_N$, the $v$-integral in \eqref{kernel0} picks up a residue at $u$, whose $u$-integral can be computed explicitly. For example, if $\Sigma_N$ and $c+i\mathbb R$ intersect at $c\pm i\alpha$ with $\alpha>0$, we obtain in this way that
\begin{equation}\label{kernelintersect0}K_N(x,x'):=\int_{\Sigma_N}\frac{\dv u}{2\pi i}\int_{c+i\mathbb R}\frac{\dv v}{2\pi i}\frac{W(v)\prod_{j=1}^N(v-a_j)}{W(u)\prod_{j=1}^N(u-a_j)}\frac{e^{-xv+x'u}}{v-u} + e^{c(y-x)}\frac{\sin\alpha(x-y)}{\pi(x-x')}.\end{equation}
If $c+i\mathbb R$ lies to the left of $\Sigma_N$, this residue contribution disappears again upon evaluation of the $u$-integral.
\end{remark}
\begin{remark}
In the fully confluent case where $a_1=\cdots=a_N=0$, we can also re-write \eqref{kernelintersect0} (again using a residue computation)
as
\begin{equation*}\label{kernelintersect}K_N(x,x'):=\oint_{\Sigma_N}\frac{\dv u}{2\pi i}\int_{c+i\mathbb R}\frac{\dv v}{2\pi i}\frac{W(v)}{W(u)}\left(\frac{v^N}{u^N}-1\right)\frac{e^{-xv+x'u}}{v-u}.\end{equation*}
It is easily seen that this is equal to
\begin{equation}\label{kieburg_kernel}
        K_N(x,x')=\oint_\mathcal C\frac{\dv u}{2\pi }\int_{-\infty}^{+\infty}\frac{\dv v}{2\pi }\frac{\mathcal F[w](v)}{\mathcal F[w](u)}\left(1-\left(\frac{v}{u}\right)^N\right)\frac{e^{i(xu-x'v)}}{u-v},
    \end{equation}
    where the contour $\mathcal C$ of $u$ is encircling $0$ in the positive direction. The latter expression was obtained in \cite[Corollary III.3]{kieburg17} for the correlation kernel of any orthogonal polynomial ensemble of additive derivative type.
\end{remark}

\begin{remark}If not all $a_1,\ldots, a_N$ are distinct and not all of them are equal, the kernel $K_N$ still defines a biorthogonal measure, and we have explicit expressions for it. See Section \ref{sec:confluent} below for details.
\end{remark}
\begin{remark}\label{remark:scaleshift}
The identification \eqref{def:Ww} between $w$ and $W$ behaves well under scaling and shifting. Suppose that $w\in\mathcal U_{c_-,c_+}^N$, with corresponding function $W\in\mathcal V_{c_-,c_+}^N$. Define the shifted and scaled weight function
\begin{equation*}w_{c_0,x_0}(x)=\frac{1}{c_0}w\left(\frac{x-x_0}{c_0}\right),\end{equation*} with $c_0\in\mathbb R\setminus\{0\}$, $x_0\in\mathbb R$, such that $w_{c_0,x_0}$ belongs to $\mathcal U_{c_-/c_0,c_+/c_0}^N$ if $c_0>0$ and to $\mathcal U_{c_+/c_0,c_-/c_0}$ if $c_0<0$.
The associated function $W_{c_0,x_0}$ is given by
\begin{equation*}W_{c_0,x_0}(z)=\int_{-\infty}^{+\infty}w_{c_0,x_0}(x)e^{xz}\mathrm d x=\int_{-\infty}^{+\infty}w(y)e^{x_0z+c_0yz_0}\mathrm d y =e^{x_0 z}W\left(c_0z\right).\end{equation*} This simple property will be useful later.
\end{remark}
\begin{remark}
The double contour integral expression \eqref{kernel0} is similar to a double contour integral expression for the correlation kernel of Schur measures \cite{Okounkov, Okounkov2}. In the latter, both integrations are over closed curves in the complex plane, which is related to the fact that the Schur measures live on a discrete lattice rather than on the real line. In \cite[Section 4]{BorodinPeche}, certain continuous versions of Schur measures were studied, with correlation kernels of the same form as \eqref{kernel0}, but in the special case where $1/W$ is a polynomial.
\end{remark}
\begin{example}\label{example:GUE+ and LUE+}
The GUE with external source corresponds to \eqref{BEDT} with $w(x)=e^{-x^2/2}$, such that $w\in\mathcal U_{c_-,c_+}^N$ for every $c_-<0<c_+$ and for every $N\in\mathbb N$. Thus, through \eqref{def:Ww}, $W$ is given by \begin{equation*}W_{\rm GUE}(z)=e^{z^2/2}\end{equation*} and belongs to $\mathcal V_{c_-,c_+}^N$ for every $c_-<0<c_+$ and for every $N\in\mathbb N$.
The $N\times N$ LUE with external source corresponds to \eqref{BEDT} with $w(x)=x^{N-1+\nu}e^{-x}\ind_{(0,+\infty)}(x)$, such that $w\in\mathcal U_{c_-,c_+}^N$
for any $c_-<0<c_+<1$.
We then have \begin{equation*}W_{\rm LUE}^{(\nu)}(z)=\frac{1}{(z-1)^{N+\nu}},\end{equation*} and $W\in\mathcal V_{c_-,c_+}^N$ for $c_-<0<c_+<1$.
Notice that for $\nu=0$, the decay of $W$ is just sufficiently strong for it to belong to $\mathcal V_{c_-,c_+}^N$.
\end{example}
\begin{example}
In the multiplicative case $a_j=j$, it is natural to change variables by setting \begin{equation*}\tilde K_N(y,y')=\frac{1}{y} K_N(\log y, \log y'),\qquad x,y>0.\end{equation*}
This modified correlation kernel is, by \eqref{kernel0}, equal to
\begin{equation*} \tilde K_N(y,y')=\int_{\Sigma_N}\frac{\dv u}{2\pi i}\int_{c+i\mathbb R}\frac{\dv v}{2\pi i}\frac{W(v)\prod_{j=1}^N(v-a_j)}{W(u)\prod_{j=1}^N(u-a_j)}\frac{y^{-v-1}(y')^{u}}{v-u}.\end{equation*}
Then the identification between $W(z)$ and $\tilde w(y)$ in \eqref{PEmult} goes via the Mellin transform,
\begin{equation}\label{def:Wwtilde}
W(z)=\mathcal M[w](z)=\int_{0}^{\infty}y^{z-1}\tilde w(y)\mathrm dy,\qquad \tilde w(y)=\frac{1}{2\pi i}\int_{i\mathbb R}y^{z}W(-z)\mathrm dz.
\end{equation}
In the LUE, JUE and C-LUE from Examples \ref{example:LUEmult}, \ref{example:JUE} and \ref{example:CLUE}, we have the following functions $W$,
\begin{align}\label{ensembles}
    &W^{(\nu)}_\text{LUE*}(s)={\Gamma(\nu+s)},\quad W^{(\mu,\nu)}_\text{JUE}(s)={B(\mu+s,\nu+1)},\quad W^{(\beta,\gamma)}_\text{C-LUE}(s)={B(\nu+s,\mu-s+1)},
\end{align}
in terms of the Gamma and Beta functions. In all three ensembles we require $\nu>-1$ for the integrability of the functions $w$. Furthermore for the weight functions $w$ of JUE and C-LUE to be integrable, it is necessary to require $\mu>N-1$. So the functions $W$ corresponding to both of these ensembles are $N$-dependent. We also remark here that compared with the expressions in~\cite{KK16}, a shift $\frac{n-1}{2}$ for $s$ has been made due to different conventions of the spherical transform, and the normalization constant has been omitted as it can be canceled out in the ratio of the double integral formula~\eqref{kernel0}.
\end{example}

\begin{example}
General P\'olya frequency density $w$ of infinite order can be characterized in terms of their Fourier transform, or in terms of the associated functions $W$ through \eqref{def:Ww}. They have the form
\begin{equation}\label{def:Polyafreq}
	W(z)=\mathcal F[w](-iz)=ce^{\tau z^2/2+\gamma z}\prod_{j=1}^\infty\frac{e^{-b_j z}}{1-b_j z},
\end{equation}
with the restrictions of parameters
\begin{equation}\label{eq:Polyaparameters}
	\tau\ge 0,\qquad\gamma\in\R,\qquad b_j\in\R, \qquad 0<\tau+\sum_{j=1}^\infty b_j^2<\infty.
\end{equation}
Note that the definition of P\'olya frequency functions in~\cite{karlin} looks slightly different from what we have, namely they have the more general form
\begin{equation}
    W(z)=\mathcal F[w](-iz)=cz^r e^{\tau z^2/2+\gamma z}\prod_{j=1}^\infty\frac{e^{-b_j z}}{1-b_j z}.
\end{equation} However, we require $w$ to be integrable, such that $W(0)=\int_{-\infty}^{+\infty}w(x)\dv x\in (0,+\infty)$, which is only possible if $r=0$. We will prove in Proposition \ref{lem_Polya_bd} that  functions of the form \eqref{eq:Polyaparameters} belong to $\mathcal V_{c_-,c_+}^N$ with $c_-=\max_{b_j<0} b_j^{-1}$ and $c_+=\min_{b_j>0} b_j^{-1}$ provided that $\tau>0$ or $\tau=0$ with at least $N$ non-zero $b_j$-values. If there are no positive $b_j$-values, we can take $c_+>0$ arbitrarily, and if there are no negative $b_j$-values, we can take $c_-<0$ arbitrarily.
As the corresponding functions $w$ give rise to non-negative biorthogonal ensembles of derivative type, this provides us with a practical method to construct such ensembles.
The special case $b_j=0$ for all $j$ is (up to a shift by $\gamma$ and scaling with $\tau$) the GUE with external source, while $\tau=0$, $b_1=\ldots=b_{N+\nu}=1$ and $b_j=0$ for $j>N+\nu$ is the LUE with external source (again, up to shifting and scaling), with an integer parameter $\nu$.
\end{example}

\subsection{Limit kernels for additive perturbations of LUE}\label{section:resultsLUE+}

For our next result, we return to polynomial ensembles of additive derivative type as given in \eqref{PE}. They have the remarkable property of being closed under additive convolution.
More precisely, suppose that $A_1$ and $A_2$ are independent $N\times N$ Hermitian random matrices whose eigenvalue distributions take the form \eqref{PE} with $w=w_1$ for $A_1$ and $w=w_2$ for $A_2$, with $w_1,w_2\in\mathcal U_{c_-,c_+}^N$.
Then, a special case of \cite[Theorem II.4]{kieburg17}
states that the eigenvalue distribution of $A_1+A_2$ is again a polynomial ensemble of derivative type \eqref{PE}, now with $w$ being the (additive) convolution of $w_1$ and $w_2$,
\begin{equation}\label{def:convolution}
w(x)=(w_1*w_2)(x):=\int_{\mathbb R}w_1(t)w_2(x-t)\mathrm d t.
\end{equation}
As we will explain in detail in Section \ref{section:sum}, this fact in combination with Theorem \ref{thm:algebraic} induces a simple transformation on the level of correlation kernels: the correlation kernels of $A_1$ and $A_2$ have the form \eqref{kernel0} with $W=W_1$ and $W=W_2$ given in terms of $w=w_1$ and $w=w_2$ by \eqref{def:Ww}, and the correlation kernel of $A_1+A_2$ has the form \eqref{kernel0} with $W=W_1W_2$. A similar property holds for polynomial ensembles of multiplicative derivative type \cite[Corollary 3.4]{KK16}.


This property enables us to study the correlation kernel in several interesting random matrix models.
\begin{example}\label{example:LUE+LUE}
Let $A_1$ be an $N\times N$ LUE matrix, such that its eigenvalue distribution is \eqref{PE} with $w(x)=w_1(x)=x^{N-1+\nu_1}e^{-x}\ind_{(0,+\infty)}(x)$ and $\nu_1\geq 0$, and let $A_2$ be an $N\times N$ LUE matrix corresponding to $w(x)=w_2(x)=x^{N-1+\nu_2}e^{-x}\ind_{(0,+\infty)}(x)$, independent of $A_1$.
The associated functions $W_1$ and $W_2$ defined by \eqref{def:Ww} 
are given by
\begin{equation*}W_1(z)=\frac{1}{(z-1)^{N+\nu_1}},\qquad W_2(z)=\frac{1}{(z-1)^{N+\nu_2}}.\end{equation*}
It follows that the eigenvalue distribution of the sum $A_1+A_2$ is the polynomial ensemble of derivative type associated to 
\begin{equation*}W(z)=W_1(z)W_2(z)=\frac{1}{(z -1)^{2N+\nu_1+\nu_2}}.\end{equation*}
This is the eigenvalue distribution of a single LUE matrix with parameter $\nu=N+\nu_1+\nu_2$. This is a classical result, see e.g.\ \cite[Section 7.3.2]{Anderson}.
\end{example}

\begin{example}\label{example:GUE+LUE}
Let $A_1$ be an $N\times N$ LUE matrix, such that its eigenvalue distribution is \eqref{PE} with $w(x)=w_1(x)=x^{N-1+\nu}e^{-x}\ind_{(0,+\infty)}(x)$ and $\nu\geq 0$, and with corresponding 
function $W_1$ given by
$W_1(z)=\frac{1}{(z-1)^{N+\nu}}$.
Let $A_2$ be an $N\times N$ GUE matrix independent of $A_1$, such that its eigenvalue distribution is \eqref{PE} with $w(x)=e^{-x^2/2}$.
The scaled GUE matrix $\tau A_2$, $\tau>0$, then has eigenvalue distribution (recall Remark \ref{remark:scaleshift}) corresponding to 
\begin{equation*}w(x)=w_{2}(x)=\frac{1}{\tau}e^{-\frac{1}{2\tau^2}x^2} ,\qquad W_{2}(z)=e^{\frac{\tau^2 z^2}{2}}.\end{equation*}
It follows that the eigenvalue distribution of $A_1+\tau A_2$ is the polynomial ensemble \eqref{PE} with $w=w_1*w_2$, and moreover it admits the correlation kernel
\begin{equation}\label{kernelGUELUE}
    K_N^{\rm LUE+GUE}(x,x')=\int_{\Sigma}\frac{\dv u}{2\pi i}\int_{c+i\mathbb R}\frac{\dv v}{2\pi i}\frac{v^N(1-u)^{N+\nu}}{u^N(1-v)^{N+\nu}}\frac{e^{\frac{\tau^2v^2}{2}-xv-\frac{\tau^2u^2}{2}+x'u}}{v-u},
\end{equation}
where $\Sigma$ is a simple positively oriented closed contour enclosing $0$, and $c>0$ is such that $c+i\mathbb R$ lies at the right of $\Sigma$. This same expression for the correlation kernel was derived in \cite[Equation (7.14)]{CC24}.
Another double contour integral expression was given in \cite[Theorem 2.3]{ClaeysKuijlaarsWang} for the additive convolution of a polynomial ensemble with a GUE matrix, but it is less explicit in the sense that the integrand contains a correlation kernel of the LUE.
\end{example}
As a generalization of Example \ref{example:GUE+LUE}, we now consider $A_1$ an $N\times N$ LUE matrix, perturbed by a random matrix $\tau A_2$, where $\tau>0$ and $A_2$
is a random Hermitian $N\times N$ matrix whose eigenvalue distribution is a polynomial ensemble of derivative type \eqref{PE}, for some function $w\in \bigcup_{N=1}^\infty\mathcal U_{c_-,c_+}^N$, with $c_-<0<c_+$. 
The non-negativity of the eigenvalue distribution implies that $w$ is a P\'olya frequency density of infinite order, with its corresponding function $W$ being of the form \eqref{def:Polyafreq} with $\tau>0$ or with $\tau=0$ and an infinite number of non-zero $b_j$-values.
Then, the eigenvalue distribution of $A_1+\tau A_2$ has {\em perturbed LUE} correlation kernel 
\begin{equation}\label{kernelLUEPolya}	K_N^{\rm pLUE}(x,x';\tau):=\int_{\Sigma}\frac{\dv u}{2\pi i}\int_{c+i\mathbb R}\frac{\dv v}{2\pi i}\frac{v^N(1-u)^{N+\nu}W\left(\tau v\right)}{u^N(1-v)^{N+\nu}W\left(\tau u\right)}\frac{e^{-xv+x'u}}{v-u},
    \end{equation}
where $\Sigma$ is a simple positively oriented closed contour enclosing $0$, and  $c+i\mathbb R$ with $c_-/\tau<c<1$ is such that it does not intersect with $\Sigma$.

We prove that this kernel admits a new hard edge limit kernel near $0$, which can be seen as a deformation of the Bessel kernel. 

\begin{theorem}\label{theorem:LUE+}
Let $c_-<0<c_+$ and let $W$ be independent of $N$ and belonging to $\mathcal V_{c_-,c_+}^N$ for every $N\in\mathbb N$.
Setting $\tau=\frac{r}{4N}$, the kernel \eqref{kernelLUEPolya} admits the scaling limit
    \begin{equation}
	\lim_{N\to\infty}\frac{1}{N}K_N^{\rm pLUE}\left(\frac{x}{4N},\frac{x'}{4N};\frac{r}{4N}\right)=\mathbb K^{\rm pBessel}(x,x';r)
    \end{equation}
    where
\begin{equation}\label{KLUE+0}
	\mathbb K^{\rm pBessel}(x,x';r):=\int_{\tilde\Sigma}\frac{\dv s}{2\pi i}\int_{c+i\mathbb R}\frac{\dv t}{2\pi i}\frac{(-s)^\nu}{(-t)^\nu}\frac{W(rt)}{W(rs)}\frac{\exp(\frac{1}{4t}-xt-\frac{1}{4s}+x's)}{t-s}.
    \end{equation}
    Here $c_-/r<c<0$, and $\tilde\Sigma$ is a positively oriented circle with center $c/3$ and radius $|c|/3$, passing through $0$, see Figure \ref{figure: contourspBessel}. The branch cut of the $\nu$-th power function is taken to be the negative real line so that both $s$ and $t$ have a branch cut at the positive real line, which does not enclose either of the contours.
\end{theorem}

\begin{remark}
For $r=0$ and $\nu\geq 0$, the limit kernel is equal to
\begin{equation}\label{eq:Besselkernel}
\mathbb K^{\rm pBessel}(x,x';0):=\int_{\tilde\Sigma}\frac{\dv s}{2\pi i}\int_{c+i\mathbb R}\frac{\dv t}{2\pi i}\frac{(-s)^\nu}{(-t)^\nu}\frac{\exp(\frac{1}{4t}-xt-\frac{1}{4s}+x's)}{t-s}.
\end{equation}
This is actually a known double contour integral expression for the hard edge Bessel kernel; see \cite{Girotti} and Appendix~\ref{appendix}.
\end{remark}

\begin{remark}
Applying the previous theorem in the case $W(z)=e^{\frac{z^2}{2}}$, we obtain that the kernel $K_N^{\rm LUE+GUE}$ from ~\eqref{kernelGUELUE} has a hard edge limit
\begin{equation}
\lim_{N\to\infty}\frac{1}{N}K_N^{\rm LUE+GUE}\left(\frac{x}{4N},\frac{x'}{4N};\frac{r}{4N}\right)=\int_{\tilde\Sigma}\frac{\dv s}{2\pi i}\int_{c+i\mathbb R}\frac{\dv t}{2\pi i}\frac{(-s)^\nu}{(-t)^\nu}\frac{\exp(\frac{r^2t^2}{2}+\frac{1}{4t}-xt+\frac{r^2s^2}{2}-\frac{1}{4s}+x's)}{t-s}.
\end{equation}
This is an alternative (and simpler) expression of a limit kernel computed in a different way in ~\cite[Eqn. (2.9)]{CD16}. The same limit kernel appeared in \cite[Formula (1.9)]{Veto} in the context of a first passage percolation problem.
\end{remark}

\begin{remark}
The above result does not require the biorthogonal ensemble of derivative type associated with the function $W$ to be non-negative. However, the probabilistic interpretation of the limit kernel as kernel of a (non-negative) determinantal point process is then no longer valid.
\end{remark}

\begin{figure}[t]
\begin{center}
		\begin{tikzpicture}		
			\node at (0,0) {};
			\fill (0,0) circle (0.05cm);
			\node at (0.15,-0.2) {$0$};
			 \draw[dashed,->-=1,black] (0,-2) to [out=90, in=-90] (0,2.2);
            \draw[gray] (-3,-2) to [out=90, in=-90] (-3,2.2);
            \draw[gray] (2,-2) to [out=90, in=-90] (2,2.2);
            \draw[black,->-=0.6] (-2.4,-2) to [out=90, in=-90] (-2.4,2.2);
			\draw[dashed,->-=1,black] (-4.4,0) to [out=0, in=-180] (4.4,0);
			\node at (-1.6,0.9) {\small $\tilde\Sigma$};
		\node at (-2.1,-2.1) {\small $c+i\mathbb R$};	
        		\node at (-3.5,-2.1) {\small $c_-+i\mathbb R$};	
                		\node at (2.2,-2.1) {\small $c_++i\mathbb R$};	
		\draw[->-=0.4,black] ([shift=(0.180:1.0cm)]-1.3,-0.7) arc (-50.180:360:.88cm);
\end{tikzpicture}
    \caption{The vertical strip $c_-<\Re z<c_+$ and the integration contours $\Sigma$ and $c+i\mathbb R$ in \eqref{KLUE+0}.}
    \label{figure: contourspBessel}
\end{center}
\end{figure}
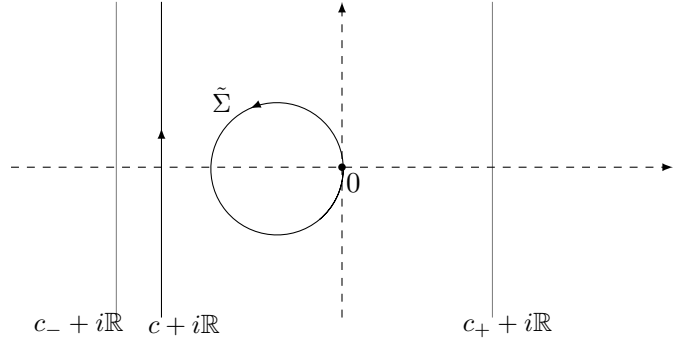

\subsection{Limit kernels for Muttalib-Borodin deformations}
\label{section:resultsMB}
Next, we consider the following ensemble for $\theta>0$ and $\eta>-\theta$:
\begin{equation}\label{MB0add}
     \frac{1}{Z_N(\theta)}\det[e^{(\theta j+\eta) x_k}]_{j,k=1}^N\det[(-\partial_{x_j})^{k-1}w(x_j)]_{j,k=1}^N \mathrm d x_1\ldots \mathrm d x_N,\quad x_1,\ldots,x_N\in\mathbb R,
\end{equation}
or equivalently in the variables $y_j=e^{x_j}$,
\begin{equation}\label{MB0mult}
     \frac{1}{Z_N(\theta)}\det[y_k^{\theta j +\eta -1}]_{j,k=1}^N\det[(-y_j\partial_{y_j})^{k-1}\tilde w(y_j)]_{j,k=1}^N \mathrm d y_1\ldots \mathrm d y_N,\quad y_1,\ldots,y_N> 0.
\end{equation}
Notice that \eqref{MB0add} is equal to the biorthogonal ensemble of derivative type \eqref{BEDT} with $a_j=\theta j +\eta$.
For $\theta=0$ and $\eta=0$, we recover the polynomial ensemble of additive derivative type \eqref{PE}, while for $\theta=1$ and $\eta=0$, \eqref{MB0mult} reduces to the polynomial ensemble of multiplicative derivative type \eqref{PEmult}.

In order to apply Theorem \ref{thm:algebraic}, we need that $w\in\mathcal U_{c_-,c_+}^N$ with $c_+>N\theta+\eta$.
Since we are interested in the large $N$ asymptotic behavior, we now assume that there exists $c_-<0$ such that $w$ belongs to the space $\mathcal U_{c_-,c_+}^N$ for every $N$ and for every $c_+>0$. We can construct non-negative biorthogonal ensembles of this type by taking $W$ of the form \eqref{def:Polyafreq} with $\tau>0$ or with $\tau=0$ and an infinite number of non-zero $b_j$-values, and moreover such that all non-zero $b_j$-values are negative.
Then, the kernel of the ensemble \eqref{MB0mult} is given by
\begin{equation}\label{kernelMBPolya}	K_N^{\rm MB}(y,y';\theta,\eta):=\oint_{\Sigma_N}\frac{\dv u}{2\pi i}\int_{c+i\mathbb R}\frac{\dv v}{2\pi i}\frac{W\left(v\right)\prod_{j=1}^N(v-\theta j-\eta )}{W\left(u\right)\prod_{j=1}^N(u-\theta j-\eta )}\frac{y^{-v-1}(y')^{u}}{v-u},
    \end{equation}
where $\Sigma_N$ is a simple positively oriented closed contour enclosing $\theta+\eta,2\theta+\eta,\ldots, N\theta+\eta$, and  $c+i\mathbb R$ with $c_-<c<\theta+\eta$ is such that it does not intersect with $\Sigma_N$.

\begin{example}\label{example:MBLaguerre}
For $\eta=1-\theta$ and $\tilde w(y)=y^{\nu} e^{-y}$, \eqref{MB0mult} is equal to
\begin{equation}\label{MBLaguerre}
     \frac{1}{Z_N(\theta)}\prod_{j<k}(y_j-y_k)(y_j^\theta-y_k^\theta)\prod_{j=1}^N
     y_j^{\nu}e^{-y_j}\mathrm d y_j.
\end{equation}
This is the Laguerre Muttalib-Borodin ensemble, see \cite{Borodin}.
\end{example}
\begin{example}\label{example:GUEequi}
For $w(x)=e^{-x^2/2}$, \eqref{MB0add} is the eigenvalue distribution of the GUE with equi-spaced external source. More general unitary invariant random matrix ensembles with equi-spaced external source have been studied in \cite{CW1,CW2}.
\end{example}

We can interpret the model \eqref{MB0mult} as a Muttalib-Borodin type deformation of a polynomial ensemble of multiplicative derivative type \eqref{PEmult}, which moreover interpolates between additive ($\theta=0$, in logarithmic variables $x_j=\log y_j$) and multiplicative ($\theta=1$) derivative type 
polynomial ensembles as $\theta$ increases from $0$ to $1$. 
It is well-known that hard edge scaling limits of Muttalib-Borodin ensembles give rise to limit kernels generalizing the hard edge Bessel kernel, called Wright's generalized Bessel kernels \cite{Borodin, Molag, Wang}. In our next result, we obtain a new class of limit kernels generalizing these Wright's generalized Bessel kernels. 


\begin{theorem}\label{thm:MB}
Let $c_-<0$ and let $W$ belong to $\mathcal V_{c_-,c_+}^N$ for all $c_+>c_-$ and for all natural numbers $N$. Assume that there exist $C,\tilde C>0$ and $c>\frac{\pi}{2\theta}$, and a region $A$ of the form \[A=\{z\in\mathbb C:\Re z\geq \theta+\eta-\epsilon, -\epsilon\leq \Im z\leq \epsilon\},\qquad \epsilon>0,\] such that 
\begin{align}
&\label{eq:MBcondW1}|W(\eta+it)|=\mathcal O\left(e^{-c|t|}\right),&\mbox{as $t\to \pm\infty$,}\\
&\label{eq:MBcondW2}\frac{1}{|W(u)|}\leq \tilde C e^{C\Re u}, &\mbox{for $u\in A$.}
\end{align}
The kernel $K_N^{\rm MB}(y,y';\theta,\eta)$ of the ensemble \eqref{MB0mult} with $\theta>0$ and $\eta>-\theta$ admits the  scaling limit
    \begin{equation}       
\lim_{N\to\infty}\frac{1}{N^{1/\theta}}K_N^{\rm MB }\left(\frac{y}{N^{1/\theta}},\frac{y'}{N^{1/\theta}};\theta,\eta\right)=\mathbb K^{\theta,\eta}(y,y'),
    \end{equation}
    where
    \begin{equation}\label{def:KMB}
        \mathbb K^{\theta,\eta}(y,y'):=\int_{\Sigma}\frac{\dv u}{2\pi i}\int_{\eta+i\mathbb R}\frac{\dv v}{2\pi i}\frac{W(v)\Gamma\left(1-\frac{u-\eta}{\theta}\right)}{W(u)\Gamma\left(1-\frac{v-\eta}{\theta}\right)}\frac{y^{-v-1}(y')^{u}}{v-u},
    \end{equation}
    with $\Sigma\subset A$ starting and ending at $+\infty$, going around $[\theta+\eta,+\infty)$ and not intersecting with the vertical line $\eta+i\R$. See Figure \ref{figure: contoursMB}.
\end{theorem}
\begin{remark}\label{remark:convergence}
Observe that the condition \eqref{eq:MBcondW1} implies in particular, by Stirling's approximation, that 
\[\frac{W(\eta+it)}{\Gamma\left(1-\frac{it}{\theta}\right)}=\mathcal O\left(|t|^{-1/2}e^{\frac{\pi}{2\theta}|t|}e^{-c|t|}\right)\qquad \mbox{as $t\to\pm\infty$,}\]
such that
$y^{-v-1}\frac{W(v)}{\Gamma\left(1-\frac{v-\eta}{\theta}\right)}$
is integrable on $\eta+i\mathbb R$.
Also, by the reflection formula and Stirling's approximation, we have
\[\frac{\Gamma\left(1-\frac{u-\eta}{\theta}\right)}{W(u)}=\frac{\pi}{W(u)\sin\left(\pi\frac{u-\eta}{\theta}\right)\Gamma\left(\frac{u-\eta}{\theta}\right)}=\mathcal O\left(\frac{e^{-\frac{u-\eta}{2\theta}\log\frac{u-\eta}{\theta}}}{W(u)|\Im u|}\right)\qquad \mbox{as $\Re u\to +\infty$},\]
such that \eqref{eq:MBcondW2} implies that $(y')^{u}\frac{\Gamma\left(1-\frac{u-\eta}{\theta}\right)}{W(u)}$
is integrable on $\Sigma$ for every $y'>0$, provided that we take $\Sigma$ in a suitable way, inside $A$ and not too close to the real line. Since $\Sigma$ and $\eta+i\mathbb R$ do not intersect, it then follows that the double integral in \eqref{def:KMB} is absolutely convergent.
\end{remark}

\begin{remark}
If $w\in\mathcal U_{c_-,c_+}^N$ for all $c_+>c_-$ and for all natural numbers $N$, and moreover it is analytic with sufficiently rapid decay in a horizontal strip $|\Im z|<c'$ with $c'>2\pi$, then condition \eqref{eq:MBcondW1} is satisfied. Indeed, in that case we can deform the $x$-integral in \eqref{def:Ww} from $\mathbb R$ to $i\sgn(y)c+\mathbb R$ with $\pi/2<c<c'$, and obtain 
\[|W(\eta +iy)|\leq \int_{i\sgn(y)c+\mathbb R}|e^{(\eta+iy)s}w(s)|\dv s\leq e^{-c|y|}\int_{\mathbb R}e^{\eta x}|w(x+i\sgn(y)c)|\dv x.\]
On the other hand, if $w(x)\geq 0$ for all $x\in\mathbb R$ and $w$ is not identically zero, then again by \eqref{def:Ww}, we have for every $a<b\in\mathbb R$ that 
\[|W(u)|\geq |\Re W(u)|=\left|\int_{-\infty}^{+\infty}e^{x\Re u}\cos(x\Im u)w(x)\dv x\right|
\geq e^{a\Re u}\int_a^{b}\cos(x\Im u)w(x)\dv x.\] It now suffices to take $a,b,A$ such that
$\cos(x\Im u)\geq 1/2$ for $x\in[a,b], u\in A$ and such that
$\int_a^{b}w(x)\dv x>0$ and we obtain \eqref{eq:MBcondW2}. 
\end{remark}
\begin{remark}
We point out that $W$ has to be independent of $N$. This excludes, for instance, the functions $W$ underlying the JUE and the C-LUE~\eqref{ensembles}. Nevertheless, we do recover some well-known random matrix ensembles and limit kernels.
\begin{itemize}
    \item For $\theta=1$, $\eta=0$ and $W(z)={\prod_{j=1}^M\Gamma(z+\nu_j)}$,~\eqref{MB0mult} is the singular value distribution of products of Ginibre matrices, and the kernel~\eqref{def:KMB} is equivalent to the Meijer-G kernel; see~\cite[Eqn. (5.7)]{KuijlaarsZhang} and~\cite[Eqn. (1.11)--(1.12) and Appendix A]{CGS19}.\smallskip
    \item For $\eta=0$, $\theta>1$ and $W(z)={\Gamma(z+\nu)}$, ~\eqref{MB0mult} reduces to the Laguerre Muttalib-Borodin ensembles, and the kernel~\eqref{def:KMB} is equal to Wright's generalized Bessel kernel, see \cite[Eqn. (1.11)--(1.14) and Appendix A]{CGS19}. Note that the condition $\theta>1$ is needed for \eqref{eq:MBcondW1} to hold.\smallskip
    \item It can be checked that for any P\'olya frequency density \eqref{def:Polyafreq} with~\eqref{eq:Polyaparameters} satisfying $\tau>0$ and all $b_j<0$, the corresponding function $W$ satisfies the assumptions of Theorem \ref{thm:MB}. Thus every such $W$ gives rise to a different hard edge kernel in the form of~\eqref{def:KMB}.
\end{itemize}
\end{remark}

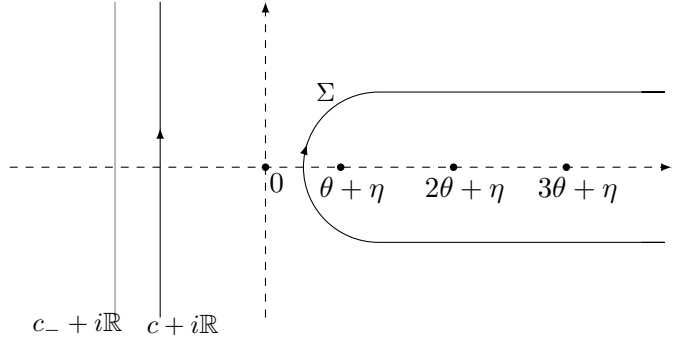
\begin{figure}[t]
\begin{center}
		\begin{tikzpicture}		
			\node at (0,0) {};
			\fill (-1,0) circle (0.05cm);
			\node at (-0.85,-0.2) {$0$};
            \fill (0,0) circle (0.05cm);
			\node at (0.15,-0.3) {$\theta+\eta$};
            \fill (1.5,0) circle (0.05cm);
            \node at (1.65,-0.3) {$2\theta+\eta$};
            \fill (3,0) circle (0.05cm);
            \node at (3.15,-0.3) {$3\theta+\eta$};
			 \draw[dashed,->-=1,black] (-1,-2) to [out=90, in=-90] (-1,2.2);
            \draw[gray] (-3,-2) to [out=90, in=-90] (-3,2.2);
                        \draw[black,->-=0.6] (-2.4,-2) to [out=90, in=-90] (-2.4,2.2);
			\draw[dashed,->-=1,black] (-4.4,0) to [out=0, in=-180] (4.4,0);
			\node at (-0.2,1) {\small $\Sigma$};
		\node at (-2.1,-2.1) {\small $c+i\mathbb R$};	
        		\node at (-3.5,-2.1) {\small $c_-+i\mathbb R$};	
                \draw[-<-=2,black] (0.5,1) to [out=0, in=0] (4,1);
                \draw[black] (0.5,-1) to [out=0, in=0] (4,-1);
                \draw[-<-=0.4,black] (0.5,1) arc (90.270:270:1cm);
\end{tikzpicture}
    \caption{The half-plane $\Re z>c_-$ and the integration contours $\Sigma$ and $c+i\mathbb R$ in \eqref{def:KMB}.}
    \label{figure: contoursMB}
\end{center}
\end{figure}

\subsection{Outline}

In Section \ref{section:alg}, we will study the biorthogonal ensembles of derivative type in detail. We will first study the relevant function spaces $\mathcal U_{c_-,c_+}^N$ and $\mathcal V_{c_-,c_+}^N$ for the functions $w$ and $W$. Secondly, we will derive the identity \eqref{ZN0} for the partition function $Z_N(a_1,\ldots, a_N)$.
Afterwards, we will construct  biorthogonal systems for biorthogonal ensembles of derivative type with distinct $a_1,\ldots, a_N$. This will enable us to prove the double contour integral formula for the correlation kernel, first in the distinct case, and later also in the general case by taking confluent limits. All these results will lead to the proof of Theorem \ref{thm:algebraic}.
At the end of Section \ref{section:alg}, we will provide more details about additive convolutions of polynomial ensembles of derivative type and about P\'olya frequency densities.

In Section \ref{section:asymp}, we will analyze special cases of biorthogonal ensembles of derivative type asymptotically. First, we will consider ensembles arising from additive perturbations of LUE and prove Theorem \ref{theorem:LUE+}. Secondly, we will consider the Muttalib-Borodin type ensembles and prove Theorem \ref{thm:MB}.

\section{Properties of biorthogonal ensembles of derivative type}\label{section:alg}

\subsection{Function spaces for \texorpdfstring{$w$}{TEXT} and \texorpdfstring{$W$}{TEXT}}

We will now show that the first relation in \eqref{def:Ww} maps the function space $\mathcal U_{c_-,c_+}^N$ for $w$ into the function space $\mathcal V_{c_-,c_+}^N$. Afterwards, we will show that the second relation in \eqref{def:Ww} maps $W\in \mathcal V_{c_-,c_+}^N$ back into $\mathcal U_{c_-,c_+}^N$ under the condition that $W$ satisfies a mild additional $L^1$ bound.

\begin{proposition}\label{prop:wtoW}
Let $w\in \mathcal U_{c_-,c_+}^N$ with $c_-<0<c_+$. The function $$W(z):=\f[w](-iz)=\int_{-\infty}^{+\infty}e^{xz}w(x)\mathrm d x.$$ belongs to $\mathcal V_{c_-,c_+}^N$.
\end{proposition}

\begin{proof}
Let $w\in\mathcal U_{c_-,c_+}^N$ for some $c_-<0<c_+$. The regularity and exponential decay of $w$ imply that $W$ is well-defined for every $z\in\mathbb C$ such that $c_-<\Re z<c_+$: indeed for sufficiently small $\epsilon>0$ and for $x\in\mathbb R$, there exists $\tilde c>0$ such that we have 
\begin{equation*}|e^{xz}w(x)|= e^{x\Re z}|w(x)|\leq \tilde c e^{-(c_+-\epsilon-\Re z)x}\ind_{[0,+\infty)]}(x)+\tilde c e^{-(c_-+\epsilon-\Re z)x}\ind_{(-\infty,0)}(x),\end{equation*}
which is integrable if $c_-+\epsilon\leq \Re z\leq c_+-\epsilon$. One readily checks using standard complex analysis arguments that $W$ is analytic in the domain $c_-<\Re z<c_+$.

Next we show integrability the decay of $W$. Using integration by parts, we obtain 
\begin{equation*}zW(z)=
-\int_{-\infty}^{+\infty}e^{xz}w'(x)\mathrm d x
,
\end{equation*}
where we notice that the boundary term vanishes because of the exponential decay of $w$. Repeating this and using the integrability and exponential decay of derivatives of $w$, we obtain for $j=0,\ldots, N-1$ that
\begin{equation}\label{def:moment}z^{j}W(z)=
(-1)^j\int_{-\infty}^{+\infty}e^{xz}w^{(j)}(x)\mathrm d x.
\end{equation}
From this expression for $j=N-1$, we immediately see that 
\begin{equation*}|z^{N-1}W(z)|\leq \int_{-\infty}^{+\infty}e^{x\Re z}|w^{(N-1)}(x)|\mathrm dx,\end{equation*}
which is uniformly bounded for $c_-+\epsilon<\Re z<c_+-\epsilon$ for any $\epsilon>0$, since there exists $\tilde c>0$ such that
\begin{equation}\label{estimatesW}e^{x\Re z}|w^{(N-1)}(x)|\leq \tilde c e^{-(c_+-\epsilon/2-\Re z)x}\ind_{[0,+\infty)]}(x)+\tilde ce^{-(c_-+\epsilon/2-\Re z)x}\ind_{(-\infty,0)}(x)\leq \tilde c e^{-\epsilon|x|/2}.\end{equation}
We conclude that \begin{equation*}\sup_{c_-+\epsilon<c<c_+-\epsilon}\|z^{N-1}W(z)\|_{L^\infty(c+i\mathbb R)}<\infty,\end{equation*}
for any $\epsilon>0$.

Next, we observe that \eqref{def:moment} expresses $z^{N-1}W(z)$ as a Fourier transform:
\begin{equation}\label{eq:WFourier}(c+it)^{N-1}W(c+it)=(-1)^{N-1}\int_{-\infty}^{+\infty}e^{ixt}e^{cx}|w^{(N-1)}(x)|\mathrm dx=(-1)^{N-1}\mathcal F[e^{c}\cdot w^{(N-1)}].\end{equation}
We know from \eqref{estimatesW} that $e^{cx}w^{(N-1)}(x)$ is integrable for every $c_-<c<c_+$, such that the Riemann-Lebesgue lemma implies that \begin{equation*}\lim_{t\to\pm\infty}(c+it)^{N-1}W(c+it)=0,\qquad \mbox{for every $c_-<c<c_+$.}\end{equation*}
By the same estimate \eqref{estimatesW}, $e^{c}\cdot w^{(j)}$ is square-integrable
with \begin{equation*}\sup_{c_-+\epsilon<c<c_+-\epsilon}\|e^{c}\cdot w^{(j)}\|_{L^2(\mathbb R)}<\infty.\end{equation*}
Since the Fourier transform is a bounded operator on $L^2(\mathbb R)$, this also implies that the function $t\in\mathbb R\mapsto (c+it)^{N-1}W(c+it)$ is also in $L^2(\mathbb R)$, so $z\in c+i\mathbb R\mapsto z^{N-1}W(z)$ is in $L^2(c+i\mathbb R)$, with 
\begin{equation*}\sup_{c_-+\epsilon<c<c_+-\epsilon}\|z^{N-1}W(z)\|_{L^2(c+i\mathbb R)}<\infty.\end{equation*}
We have now proved that $W\in \mathcal V_{c_-,c_+}^N$.
\end{proof}

One observes that $w\in\mathcal U_{c_-,c_+}^N$ does not imply the $L^1$-bound
\begin{equation*}\sup_{c_-+\epsilon<c<c_+-\epsilon}\|z^{N-1}W(z)\|_{L^1(c+i\mathbb R)}<\infty.\end{equation*}
Indeed, for the LUE with external source with $\nu=0$, we have $W(z)\sim \frac{1}{z^{N}}$ as $z\to \pm i\infty$, such that $z^{N-1}W(z)$ is not integrable on a vertical line $c+i\mathbb R$.
In general, it might happen that the second relation in \eqref{def:Ww} does not map $W\in\mathcal V_{c_-,c_+}^N$ back onto $w\in\mathcal U_{c_-,c_+}^N$. In order to ensure this, we need an additional $L^1$-condition.

\begin{proposition}\label{prop:Wtow}
    Let $W\in \mathcal V_{c_-,c_+}^N$ with $c_-<0<c_+$, and suppose in addition that
    \begin{equation}\label{eq:L1}\sup_{c_-+\epsilon<c<c_+-\epsilon}\|z^{N-1}W(z)\|_{L^1(c+i\mathbb R)}<\infty.\end{equation}
    The function $$w(x)=\frac{1}{2\pi i}\int_{c+i\mathbb R}e^{-xz}W(z)\mathrm d z$$ is independent of $c_-<c<c_+$ and belongs to $\mathcal U_{c_-,c_+}^N$.
\end{proposition}
\begin{proof}
The analyticity of $W$ in $c_-<\Re z<c_+$ together with the $L^1$ condition \eqref{eq:L1} implies that $e^{-xz}W(z)$ is integrable on $-c+i\mathbb R$ for every $c_-<c<c_+$: indeed for every such $c$ there is $M>0$ such that
\begin{equation*}
    |e^{-xz}W(z)|\leq M e^{-cx}\ind_{[-c-i,-c+i]}(z)+ e^{-cx}|z^{N-1}W(z)|,
\end{equation*}
and the right hand side is integrable on $c+i\mathbb R$.
In particular, $w(x)$ is well-defined for any $x\in\mathbb R$. 

Given $x\in\mathbb R$ and $c_-<c<c_+$, let us now integrate the function $e^{-xz}W(z)$ over the rectangle with corners $-iR, c-iR, c+iR, +iR$. By analyticity, the integral over this rectangle is zero. Moreover, as $R\to +\infty$, the contribution of the segments $[-iR,c-iR]$ and $[iR,c+iR]$ converges to $0$ by the dominated convergence theorem. Indeed,
$\lim_{t\to\pm\infty}W(c+it)=0$, and $W$ is uniformly bounded on $[\pm iR,c\pm iR]$ since
\begin{equation*}\sup_{c_-+\epsilon<c<c_+-\epsilon}\|z^{N-1}W(z)\|_{L^\infty(c+i\mathbb R)}<\infty.\end{equation*}
It follows that $w(x)=\frac{1}{2\pi i}\int_{-c+i\mathbb R}e^{xz}W(-z)d x$ is independent of $c_-<c<c_+$. 

It is also straightforward to check that we can bring derivatives inside the integral:
\begin{equation*}w^{(j)}(x)=\frac{1}{2\pi i}\int_{c+i\mathbb R}e^{-xz}(-z)^{j}W(z)\mathrm d z,\qquad j=0,\ldots, N-1,\end{equation*}
for every $c_-<c<c_+$.
Parametrizing $c+i\mathbb R$, we get
\begin{equation}\label{eq:wjFourier}w^{(j)}(x)=\frac{e^{-cx}}{2\pi}\int_{-\infty}^{+\infty}e^{-ixt}(-c-it)^{j}W(c+it)\mathrm d t,\qquad j=0,\ldots, N-1.\end{equation}
Using again the $L^1$-bound \eqref{eq:L1}, we see that the remaining integral on the right is bounded, which yields the exponential decay condition of $w^{(j)}$. Since 
\begin{equation*}
    \sup_{c_-+\epsilon<c<c_+-\epsilon}\|z^{N-1}W(z)\|_{L^2(c+i\mathbb R)}<\infty,
\end{equation*}
it also follows from the Fourier integral representation \eqref{eq:wjFourier} that $e^{cx}w^{(N-1)}(x)$ has uniformly bounded $L^2(\mathbb R)$-norm for $c_-+\epsilon<c<c_+-\epsilon$.
For $x\geq 0$, we write
\begin{equation*}w^{(j)}(x)= e^{-(c_+-\epsilon)x}. e^{(c_+-\epsilon)x}w^{(j)}(x),\end{equation*}
while for $x<0$, we write
\begin{equation*}w^{(j)}(x)= e^{-(c_-+\epsilon)x}. e^{(c_-+\epsilon)x}w^{(j)}(x).\end{equation*}
In either case, we factorized $w^{(j)}$ into a product of two square-integrable functions, such that  $w^{(j)}$ is integrable by the Cauchy-Schwarz inequality. This completes the proof that $w\in\mathcal U_{c_-,c_+}^N$.
\end{proof}

\subsection{Partition function}
Next, we will compute the partition function $Z_N(a_1,\ldots, a_N)$ defined by \eqref{def:ZN}.
\begin{proposition}\label{prop:partition}
Let $w\in \mathcal U_{c_-,c_+}^N$ and let $a_1,\ldots, a_N\in\mathcal S_{c_-,c_+}$.
Then, the partition function \eqref{def:ZN} is  given by
    \begin{equation}\label{ZN}
Z_N(a_1,\ldots,a_N)=N!\prod_{j=1}^NW(a_j),
    \end{equation}
    where $W$ is given by \eqref{def:Ww}.
If $\int_{\mathbb R}w(x)\mathrm d x>0$ and $w$
is a non-negative $C^{N-1}(\R)$ function satisfying
\begin{equation}\label{PFF}
    \frac{\det[(-\partial_{x_j})^{k-1}w(x_j)]_{j,k=1}^N}{\Delta(x)}\ge 0,
\end{equation}
then $W(x)>0$ for $c_-<x<c_+$.  Moreover, we have in this case that for any $c_-<a_1,\ldots, a_N<c_+$, \eqref{BEDT} with $Z_N(a_1,\ldots, a_N)$ given by \eqref{ZN} defines a probability measure on $\mathbb R^N$.  
\end{proposition}
\begin{proof}
    If $a_1,\ldots, a_N$, we use \eqref{def:ZN} and apply Andreief's identity (see, e.g.~\cite{Forrester19}) to obtain
    \begin{equation}\label{2.6}	Z_N(a_1,\ldots,a_N)=N!\frac{\det\left[\int_\R e^{a_jx}(-\partial_{x})^{k-1}w(x)\dv x\right]_{j=1}^k}{\Delta(a)}.
    \end{equation}
    One can then integrate by parts, the boundary terms vanish, and one obtains
    \begin{equation}
        \int_\R e^{a_jx}(-\partial_{x})^{k-1}w(x)\dv x=a_j^{k-1}\int_\R e^{a_jx}w(x)\dv x=a_j^{k-1}W(a_j).
    \end{equation}
    Substituting this in~\eqref{2.6} we find
        \begin{equation}\label{2.6b}	Z_N(a_1,\ldots,a_N)=N!\frac{\Delta(a)\prod_{j=1}^NW(a_j)}{\Delta(a)}=N! \prod_{j=1}^NW(a_j).
    \end{equation}
By continuity of the right hand side in $a_1,\ldots, a_N$, $Z_N(a_1,\ldots, a_N)$
 is also defined if $a_1,\ldots, a_N$ are not distinct.

    Now, suppose that $w$ satisfies~\eqref{PFF} and that $a_1,\ldots, a_N$ are real-valued. Since $w\ge 0$ and $w$ is not identically $0$, we see that 
    \begin{equation}
W(a_j)=\int_\R e^{a_jx}w(x)\dv x\ >0\quad \mbox{such that}\quad Z_N(a_1,\ldots, a_N)>0.
    \end{equation}
    Combining this with~\eqref{PFF} and~\eqref{ZN}, we see that~\eqref{BEDT} is non-negative if and only if
    \begin{equation}
\frac{\det[e^{a_jx_k}]_{j,k=1}^N}{\Delta(x)\Delta(a)}\ge 0.
    \end{equation}
    Notice that the left-hand side is equal to the Harish-Chandra-Itzykson-Zuber \cite{HC, IZ} integral \begin{equation*}\int_{U(N)}\exp\left({\rm Tr}(AUBU^*)\right)\mathrm d U\end{equation*} 
    multiplied by a positive factor. 
    Here, the integral is with respect to the Haar measure over the unitary group $U(N)$, $A$ is a Hermitian $N\times N$ matrix with eigenvalues $a_1,\ldots, a_N$, and $B$ is a Hermitian $N\times N$ matrix with eigenvalues $x_1,\ldots, x_N$.
    Since the Harish-Chandra-Itzykson-Zuber integral is non-negative, we obtain that the biorthogonal measure \eqref{BEDT} is non-negative. By the normalization, it is a probability measure.
\end{proof}

\subsection{Correlation kernel}

Next, we will construct a biorthogonal system and a correlation kernel for the ensemble \eqref{BEDT}. In other words, we search for functions $\phi_1,\ldots, \phi_N$ and $\psi_1,\ldots, \psi_N$ such that $\phi_1(x),\ldots, \phi_N(x)$ have the same linear span as $e^{a_1x},\ldots, e^{a_N x}$, $\psi_1,\ldots, \psi_N$ have the same linear span as $w, w',\ldots, w^{(N-1)}$, and they satisfy the biorthogonality relations \eqref{biorth}. Then we know that \eqref{BEDT} admits a correlation kernel $K_N$ given by \eqref{kernel}. By \eqref{biorth}, the kernel then has the reproducing property
\begin{equation*}\int_{\mathbb R}K_N(x,t)K_N(t,x')\mathrm d t=K_N(x,x'),\end{equation*}
and the correlation functions are given by \eqref{correlationfunction}; in particular \eqref{BEDT} is equal to \begin{equation*}\frac{1}{N!}\rho_N(x_1,\ldots, x_N)=\frac{1}{N!}\det\left(K_N(x_m,x_k)\right)_{m,k=1}^N.\end{equation*}

\begin{proposition}[Biorthogonal system and kernel]\label{prop:kernel}
Let $w\in\mathcal U_{c_-,c_+}^N$ and let $W$ be given by \eqref{def:Ww}. If $a_1,\ldots, a_N\in\mathcal S_{c_-,c_+}$ are distinct and $W(a_j)\neq 0$ for all $j=1,\ldots, N$, a biorthogonal system associated to \eqref{BEDT} is given by
    \begin{equation}\label{bio}
	\phi_n(x)=e^{a_nx},\quad \psi_m(x)=\frac{1}{\widetilde W_N'(a_m)}\int_{c+i\mathbb R}\frac{\dv v}{2\pi i}\frac{\widetilde W_N(v) e^{-xv}}{v-a_m}.
    \end{equation}	
    where $\widetilde W_N(z):=W(z)\prod_{j=1}^N(z-a_j)$, $c_-<c<c_+$ and $c\neq \Re a_j$ for all $j=1,\ldots, N$.
    The associated kernel \eqref{kernel} can be written in double contour integral form,
    \begin{equation}\label{kernel3}
K_N(x,x'):=\int_{\Sigma}\frac{\dv u}{2\pi i}\int_{c+i\mathbb R}\frac{\dv v}{2\pi i}\frac{\widetilde W_N(v)}{\widetilde W_N(u)}\frac{e^{-xv+x'u}}{v-u},
    \end{equation}
where $\Sigma$ is a closed positively oriented curve without self-intersections, lying in the strip $c_-<\Re z<c_+$, and enclosing $a_1,\ldots,a_N$ without enclosing zeros of $W$.
\end{proposition}
\begin{proof}
    It was proved in ~\cite[Proposition 2.1]{CC24} that~\eqref{bio} satisfies the biorthogonality relations \eqref{biorth}, and that the correlation kernel defined by~\eqref{kernel} has the double contour integral expression \eqref{kernel0}. 
 
    It remains to show that $\{\psi_m\}_{m=1,\ldots,N}$ has the same linear span as $\{(-\partial)^{m-1}w)\}_{m=1,\ldots,N}$. Note that $\{\psi_m(x)\}_{m=1,\ldots,N}$ spans all functions of the form
    \begin{equation}
	\int_{c+i\mathbb R}\frac{\dv v}{2\pi i}\frac{P(v) \widetilde W_N(v)e^{-xv}}{\prod_{j=1}^N(v-a_j)}=\int_{c+i\mathbb R}\frac{\dv v}{2\pi i}{P(v) W(v)e^{-xv}},
    \end{equation}
    where $P$ is a polynomial with $\deg P\le N-1$. On the other hand, we have
    \begin{equation}
	(-\partial_{x})^{m-1}w(x)=\int_{c+i\mathbb R}\frac{\dv v}{2\pi i}v^{m-1} W(v)e^{-xv},
    \end{equation}
    such that $\{(-\partial)^{m-1}w)\}_{m=1,\ldots,N}$ clearly spans the same vector space.
\end{proof}

\subsection{The confluent case}\label{sec:confluent}
Observe that the kernel \eqref{kernel} is also defined when $a_1,\ldots, a_N$ are not distinct, while the biorthogonal ensemble \eqref{BEDT} is not, because $\Delta(a)$ vanishes. We already showed that the ratio $\frac{\det[e^{a_jx_k}]_{j,k=1}^N}{\Delta(a)}$ admits to take confluent limits where the $a_j$-values are not all distinct. 
We will now explicitly compute such confluent limits of \eqref{BEDT}. The construction below is strongly similar to the classical construction of Lagrange interpolating polynomials with nodes of multiplicity greater than one. 

Suppose that for some $b_1,\ldots,b_p\in\R$ and $1\le N_1,\ldots,N_p\le N$, we have
\begin{equation}\label{def:confluent}
    b_1=a_1=\ldots=a_{N_1},\quad b_j=a_{N_{1}+\ldots+N_{j-1}+1}=\ldots=a_{N_{1}+\ldots+N_j} \quad\text{and}\quad N_1+\ldots+N_p=N.
\end{equation} 
When taking the partition $(N_1,\ldots, N_p)$ equal to $(1,1,\ldots,1)$, with $p=N$, we recover the distinct case, while $p=1$ and the partition $(N)$ yields the fully confluent case where all $a_i$ are equal. We can approach such a confluent situation by first taking $a_1,\ldots, a_N$ distinct, and then taking the limit where $a_{N_1+\ldots + N_{j-1}},\ldots, a_{N_1+\ldots+N_j}\to b_j$ for $j=1,\ldots, p$.
We will show that 
\begin{equation}\label{def:Pconfluent}\frac{1}{Z_N(a_1,\ldots, a_N)\Delta(a)}\det[e^{a_jx_k}]_{j,k=1}^N\det[(-\partial_{x_j})^{k-1}w(x_j)]_{j,k=1}^N\end{equation}
converges to
\begin{multline}\label{exprPconfluent}P(x_1,\ldots, x_N;b_1,\ldots, b_P;N_1,\ldots, N_p)\\:=\frac{1}{\tilde Z_N(b_1,\ldots,b_p;N_1,\ldots,N_p)}\det\left[x_k^{n_j-1}e^{b_jx_k}\right]_{\substack{j=1,\ldots,p;\ n_j=1,\ldots,N_j\\k=1,\ldots,N}}\det[(-\partial_{x_j})^{k-1}w(x_j)]_{j,k=1}^N,
\end{multline}
with 
\begin{equation}\label{partitionconfluent}
\tilde Z_N(b_1,\ldots,b_p;N_1,\ldots,N_p)=\int_{\mathbb R^N}\det\left[x_k^{n_j-1}e^{b_jx_k}\right]_{\substack{j=1,\ldots,p;\ n_j=1,\ldots,N_j\\k=1,\ldots,N}}\det[(-\partial_{x_j})^{k-1}w(x_j)]_{j,k=1}^N\mathrm d x_1\ldots\mathrm d x_N.
\end{equation}
Explicitly, the first determinant is written as the determinant of a block matrix, with each block having a Vandermonde-like structure
\begin{equation}
    \det\left[x_k^{n_j-1}e^{b_jx_k}\right]_{\substack{j=1,\ldots,p;\ n_j=1,\ldots,N_j\\k=1,\ldots,N}}:=\det
    \begin{bmatrix}
    \begin{array}{@{}c}
        \begin{matrix}
e^{b_1x_1}&e^{b_1x_2}&\ldots&e^{b_1x_N}\\        x_1e^{b_1x_1}&x_2e^{b_1x_2}&\ldots&x_Ne^{b_1x_N}\\
\vdots&\vdots&\ddots&\vdots\\
        x_1^{N_1-1}e^{b_1x_1}&x_2^{N_1-1}e^{b_1x_2}&\ldots&x_N^{N_1-1}e^{b_1x_N}
        \end{matrix}\vspace{.25em}\\
        \hline\vspace{-.5em}\\
        \vdots\vspace{.25em}\\
        \hline\vspace{-.5em}\\
        \begin{matrix}
        e^{b_px_1}&e^{b_px_2}&\ldots&e^{b_px_N}\\
        x_1e^{b_px_1}&x_2e^{b_px_2}&\ldots&x_Ne^{b_px_N}\\
        \vdots&\vdots&\ddots&\vdots\\
        x_1^{N_p-1}e^{b_px_1}&x_2^{N_p-1}e^{b_px_2}&\ldots&x_N^{N_p-1}e^{b_px_N}
        \end{matrix}
    \end{array}
    \end{bmatrix}.
\end{equation}
\begin{proposition}
    Let $N\in\mathbb N$, let $w\in\mathcal U_{c_-,c_+}^N$, and let $W\in\mathcal V_{c_-,c_+}^N$ be given by \eqref{def:Ww}. Let
    $P(x_1, \ldots, x_N;a_1,\ldots, a_N)$ be the joint probability density function given by \eqref{def:Pconfluent} if $a_1,\ldots, a_N\in\mathcal S_{c_-,c_+}$ are distinct, and by \eqref{exprPconfluent} if they are not distinct.
    Then, $P(x_1, \ldots, x_N;a_1,\ldots, a_N)$ is continuous in $(a_1,\ldots, a_N)\in \mathcal S_{c_-,c_+}^N$, and in the confluent limit \eqref{def:confluent}, it is equal to
    \begin{equation*}
        P(x_1,\ldots, x_N;b_1,\ldots, b_p;N_1,\ldots, N_p)=\frac{1}{N!}\det\left(K_N(x_j,x_k)\right)_{j,k=1}^N,
\end{equation*}
where $K_N$ is given by~\eqref{kernel0}, and $P$ is defined by \eqref{exprPconfluent}--\eqref{partitionconfluent}. The partition function defined by \eqref{exprPconfluent} has the explicit form    
\begin{equation}
	\tilde Z_N(b_1,\ldots,b_p;N_1,\ldots,N_p)=N!\prod_{j=1}^p(N_j-1)!\prod_{1\le j<k\le p}(b_k-b_j)^{N_jN_k}\prod_{j=1}^p W(b_j)^{N_j}.
    \end{equation}
\end{proposition}
\begin{proof}
We know from Proposition \ref{prop:partition} that 
\begin{equation}\label{eq:confluent1}\frac{1}{N!}\det\left(K_N(x_j,x_k)\right)_{j,k=1}^N=\frac{1}{N!\Delta(a)\prod_{j=1}^NW(a_j)}\det[e^{a_jx_k}]_{j,k=1}^N\det[(-\partial_{x_j})^{k-1}w(x_j)]_{j,k=1}^N\end{equation}
for distinct $a_1,\ldots, a_N$, and it is clearly also continuous for $a_1,\ldots, a_N\in\mathcal S_{c_-,c_+}$ distinct. Thus, to prove the continuity of $P(x_1, \ldots, x_N;a_1,\ldots, a_N)$, it is sufficient to prove that the right hand side converges in the confluent limit to \eqref{exprPconfluent}, and that the left hand side is continuous in $a_1,\ldots, a_N\in\mathcal S_{c_-,c_+}$.

By Taylor expanding the rows in the matrix $[e^{a_jx_k}]_{j,k=1}^N$ and applying row operations to subtract terms appearing in multiple rows, it is straightforward to show that
    \begin{equation}
	\lim_{a_{N_1+\ldots+N_{j-1}+1},\ldots,a_{N_1+\ldots+N_{j}}\to b_j}\frac{\det[e^{a_jx_k}]_{j,k=1}^N}{\Delta(a)}=\frac{\det[\partial_{b_j}^{n_j-1}e^{b_jx_k}]_{\substack{j=1,\ldots,p;\ n_j=1,\ldots,N_j\\k=1,\ldots,N}}}{\det[\partial_{b_j}^{n_j-1}b_j^{k-1}]_{\substack{j=1,\ldots,p;\ n_j=1,\ldots,N_j\\k=1,\ldots,N}}}.
    \end{equation}
    This is equal to
    \begin{equation*}\lim_{a_{N_1+\ldots+N_{j-1}+1},\ldots,a_{N_1+\ldots+N_{j}}\to b_j}\frac{\det[e^{a_jx_k}]_{j,k=1}^N}{\Delta(a)}=\frac{1}{C}
    \det\left[x_k^{n_j-1}e^{b_jx_k}\right]_{\substack{j=1,\ldots,p;\ n_j=1,\ldots,N_j\\k=1,\ldots,N}},\end{equation*}
 with 
\begin{equation*}
    C=\det\left[\partial_x^{n_j}x^{k-1}\big|_{x=b_j}\right]_{\substack{j=1,\ldots,p;\ n_j=1,\ldots,N_j\\k=1,\ldots,N}}.
\end{equation*}
This is a confluent Vandermonde determinant; see a recent review~\cite{confVand} and also~\cite[Eqn. (1.174)]{Forrester_book}. It was shown in~\cite{Aitken} that
\begin{equation*}
    C=\prod_{1\leq j<k\leq p}(b_k-b_j)^{N_jN_k}\ \prod_{j=1}^p(N_j-1)!
\end{equation*}
Substituting this when taking the confluent limit of \eqref{eq:confluent1}, we obtain
\begin{multline*}P(x_1,\ldots, x_N;b_1,\ldots, b_p;N_1,\ldots, N_p)\\=\frac{1}{CN!\ \prod_{j=1}^pW(b_j)^{N_j}}\det\left[x_k^{n_j-1}e^{b_jx_k}\right]_{\substack{j=1,\ldots,p;\ n_j=1,\ldots,N_j\\k=1,\ldots,N}}\det[(-\partial_{x_j})^{k-1}w(x_j)]_{j,k=1}^N.\end{multline*}
This implies the expression for $P$ given in \eqref{exprPconfluent} and the explicit expression for the partition function.

To prove that the left hand side of \eqref{eq:confluent1} is continuous in $a_1,\ldots, a_N\in\mathcal S_{c_-,c_+}$, it is sufficient to show continuity of $K_N(x,x')$. For that, we need to be able to bring confluent limits inside the double integral \eqref{kernel0}.
This follows in a straightforward manner from Lebesgue's dominated convergence theorem, since the integrand of \eqref{kernel0} is bounded in absolute value by
\begin{equation}\label{Kbound}\frac{|W(v)|\prod_{j=1}^N|v-a_j|e^{-x\re v+x'\re u}}{\inf_{u\in\Sigma_N,v\in c+i\R} |W(u)\prod_{j=1}^N(u-a_j) (v-u)|}.
\end{equation}
In a confluent limit where some of the $a_j$'s merge, we can leave the contour $\Sigma$ unchanged, so that the infimum in the denominator is positive and independent of $a_j$'s. For the same reason, $\prod_{j=1}^N|v-a_j|$ is bounded from above, uniformly in $a_j$'s, and together with $\frac{1}{v-u}$ it gives a dominating contribution of $v^{N-1}$. Since $\|v^{N-1}W(v)\|_{L^1(c+i\R)}<\infty$ we see that the right hand side of~\eqref{Kbound} is integrable.
\end{proof}

In the fully confluent case where $a_1=\ldots=a_N=a$ (and $p=1$), the biorthogonal ensemble with density \eqref{exprPconfluent} is given explicitly by the polynomial ensemble 
\begin{equation}\label{eq:fullyconfluent}
        \frac{1}{\left(\prod_{j=1}^{N}j!\right)W(a)^N}e^{a(x_1+\ldots+x_N)}\Delta(x)\det[(-\partial_{x_j})^{k-1}w(x_j)]_{j,k=1}^N\dv x_1\ldots\dv x_N,
    \end{equation}
    with the correlation kernel
    \begin{equation}
	K_N(x,y)=\int_{\Sigma}\frac{\dv u}{2\pi i}\int_{c+i\mathbb R}\frac{\dv v}{2\pi i}\frac{(v-a)^NW(v)}{(u-a)^NW(u)}\frac{e^{-xv+yu}}{v-u}.
    \end{equation}
    Here $\Sigma$ encloses $a$ and $c+i\mathbb R$ lies to the right of $\Sigma$. 
    In the case $a=0$ this brings us back to the polynomial ensemble of derivative type \eqref{PE}.
    
\subsection{Proof of Theorem \ref{thm:algebraic}}
We now have all the ingredients to complete the proof of Theorem \ref{thm:algebraic}. For part (1) of the result, let $w\in\mathcal U_{c_-,c_+}^N$. By Proposition \ref{prop:wtoW}, $W\in\mathcal V_{c_-,c_+}^N$, while Proposition \ref{prop:partition} yields \eqref{ZN0}. The expression for the correlation kernel follows from Proposition \ref{prop:kernel}.
For part (2) of the result, given $W$ satisfying the assumptions, we know from Proposition \ref{prop:Wtow} that $w\in\mathcal U_{c_-,c_+}^N$. For distinct $a_1,\ldots, a_N$, the biorthogonal functions $\phi_n,\psi_m$ in \eqref{bio} are well-defined, and so is the kernel \eqref{kernel}, which is the kernel of \eqref{BEDT} by Proposition \ref{prop:kernel}. In the fully confluent case $a_1=\ldots=a_N=0$, the result follows from \eqref{eq:fullyconfluent}.

\subsection{Additive convolutions}\label{section:sum}
We now give more details about additive convolutions of polynomial ensembles of derivative type, as announced in Section \ref{section:resultsLUE+}.
Let $A,B$ be independent Hermitian $N\times N$ random matrices with eigenvalue distribution of $A$ the polynomial ensemble of derivative type \eqref{BEDT} associated to $w_1\in\mathcal U_{c_-,c_+}^N$ 
and eigenvalue distribution of $B$ the polynomial ensemble of derivative type \eqref{BEDT} associated to $w_2\in\mathcal U_{c_-,c_+}^N$. By~\cite[Thm II.4]{kieburg17} the eigenvalue distribution of the random matrix $X=A+B$ is the polynomial ensemble \eqref{PE} associated to functions
\begin{equation*}w_1*w_2, w_1*\partial w_2, w_1*\partial^2 w_2,\ldots, w_1*\partial^{N-1}w_2.\end{equation*}
Using the definition of the convolution \eqref{def:convolution}, we easily see that these functions are equal to 
\begin{equation*}w_1*w_2, \partial(w_1* w_2), \partial^2(w_1* w_2),\ldots, \partial^{N-1}(w_1*w_2),\end{equation*}
such that the eigenvalue distribution of $A+B$ is the polynomial ensemble of derivative type \eqref{BEDT} associated to $w=w_1*w_2$. It is straightforward to verify that $w\in\mathcal U_{c_-,c_+}^N$.
The associated function $W$ defined by \eqref{def:Ww} is given by
\begin{equation*}W(z)=\mathcal F[w_1*w_2](-iz)=\mathcal F[w_1](-iz)\mathcal F[w_2](-iz)=W_1(z)W_2(z). 
\end{equation*}
We easily check that $W\in\mathcal V_{c_-,c_+}^N$: it is analytic in the strip $\mathcal S_{c_-,c_+}$ as product of two analytic functions, and it inherits the required decay properties from $W_1$ and $W_2$.
By the Cauchy-Schwarz inequality, we also obtain that 
\begin{equation*}\sup_{c_-+\epsilon<c<c_+-\epsilon}\|z^{N-1}W(z)\|_{L^1(c+i\mathbb R)}
\leq \sup_{c_-+\epsilon<c<c_+-\epsilon}\|z^{N-1}W_1(z)\|_{L^2(c+i\mathbb R)}\|W_2\|_{L^2(c+i\mathbb R)}
<\infty,\end{equation*}
We can now apply part (2) of Theorem \ref{thm:algebraic} to conclude that $K_N$ given in \eqref{kernel0} with $W=W_1W_2$ and $a_1=\ldots=a_N=0$ is the kernel of the polynomial ensemble of derivative type \eqref{PE}.

\subsection{P\'olya frequency functions}

As discussed at the end of Section \ref{section:resultsalg},
if $w$ is a P\'olya frequency density, then the corresponding function $W$, through \eqref{def:Ww}, will be of the general form \eqref{def:Polyafreq}; see~\cite{karlin}. We now prove that such functions $W$ belong to our function space $\mathcal V_{c_-,c_+}^N$ for suitable values of $c_-,c_+,N$.
\begin{proposition}\label{lem_Polya_bd}
    Let $W$ be of the form
    \begin{equation}\label{PFFW}
        W(z)=e^{\tau z^2+\gamma z}\prod_{j=1}^\infty\frac{e^{-b_jz}}{1-b_jz},
    \end{equation}
    with the restrictions of parameters
    \begin{equation}
	\tau\ge0,\quad b_j\in\R,\quad 0<\tau+\sum_{j=1}^\infty b_j^2<+\infty.
    \end{equation}
    If $\tau>0$ or $\tau=0$ with at least $N$ non-zero $b_j$-values, $W$ belongs to the space $\mathcal V_{c_-,c_+}^N$, with $c_-=\max_{b_j<0} b_j^{-1}$ and $c_+=\min_{b_j>0} b_j^{-1}$; if there are no negative (positive) $b_j$-values, we may take $c_-$ ($c_+$) arbitrary.
\end{proposition}
\begin{proof}
Note first that $b_j\to 0$ as $j\to\infty$, otherwise $\sum_{j=1}^\infty b_j^2$ cannot be finite.
Thus, for any compact subset $U$ of $c_-<\Re z<c_+$, there exists $k$ such that $\sup_{z\in U}|b_j z|<1$ for all $j>k$, and hence by the Taylor expansion, we have for $j>k$,
\begin{equation*}\sup_{z\in U}\left|\log \frac{e^{-b_jz}}{1-b_j z}\right|\leq C b_j^2,\end{equation*}
for some constant $C>0$.
This implies that the series
\begin{equation*}\log W(z)=\tau z^2+\gamma z +\sum_{j=1}^\infty \log \frac{e^{-b_j z}}{1-b_j z}\end{equation*} converges uniformly on $U$. As a uniform limit of analytic functions, $\log W$ is analytic in $U$, and thus in the whole sector $c_-<\Re z<c_+$. The same holds for $W$.

Next, we observe that $\frac{e^{-b_jz}}{1-b_j z}$ is uniformly bounded in $c_-+\epsilon<\Re z<c_+-\epsilon$ and $\mathcal O(z^{-1})$ as $z\to\infty$.
Thus, if there are at least $n$ non-zero $b_j$-values (say, after rearranging, $b_1,\ldots, b_n$), we have
\begin{equation*}|W(z)|\leq \frac{C}{z^n} e^{\tau \Re(z^2)} \prod_{j=n+1}^{\infty}\frac{e^{-b_j z}}{1-b_j z},\end{equation*}
for $z$ large. 
The remaining infinite product is uniformly bounded.
We now suppose either $n\geq N$ or $\tau>0$.
From this, we easily infer that, for $p=2$ and for $p=\infty$, 
\begin{equation*}\sup_{c_-+\epsilon<c<c_+-\epsilon}\|z^{N-1}W(z)\|_{L^p(c+i\mathbb R)}<\infty,\end{equation*}
and that 
\begin{equation*}\lim_{t\to \pm\infty}(c+it)^{N-1}W(c+it)=0,\end{equation*}
for every $c_-<c<c_+$.
\end{proof}

\section{Asymptotics of biorthogonal ensemble of derivative type}\label{section:asymp}
\subsection{General considerations}
We will now exploit the double contour integral expression for the correlation kernel to derive large $N$ asymptotics for biorthogonal ensembles of derivative type.

Recall from Theorem \ref{thm:algebraic} that \eqref{BEDT} admits the correlation kernel
\eqref{kernel0}.
By re-scaling the integration variables with a factor $N$, we obtain 
\begin{equation}
	\frac{1}{N}K_N(x,x')=\oint_{\frac{1}{N}\Sigma_N}\frac{\dv u}{2\pi i}\int_{\frac{c}{N}+i\mathbb R}\frac{\dv v}{2\pi i}\frac{e^{N\phi_N(v;x)-N\phi_N(u;x')}}{v-u},
\end{equation}
where $\frac{1}{N}\Sigma_N$ encloses the points $a_1/N,\ldots, a_N/N$, and where the phase function $\phi_N$ is given by
\begin{equation*}\phi_N(z;x):=-xz + \frac{1}{N}\log W(Nz) + \frac{1}{N}\sum_{j=1}^N\log \left(z-\frac{a_j}{N}\right).\end{equation*}
This expression, together with the observations made in Remark \ref{remark:contourdef} allowing us to move the vertical line $\frac{c}{N}+i\mathbb R$ through the closed loop $\Sigma_N$, provides us with a convenient starting point for asymptotic analysis via saddle point methods. 

The key for a successful asymptotic analysis of the kernel $K_N$ lies in a good understanding of the phase function $\phi_N(z;x)$ and in particular of its saddle points, which are given by the equation
\begin{equation}\label{eq:saddlepoint}
    (\log W)'(z)+\sum_{j=1}^N\frac{1}{Nz-a_j}=x.
\end{equation}
Of course, as usual in saddle point analysis, the success of the method depends on the possibility to deform integration contours to regions where the real part of $\phi_N$ is either positive or negative, which requires an ad-hoc case-by-case analysis. 

Nevertheless, and this is the message we want to convey here, the double contour integral expression allows one in principle to compute the macroscopic limit density $\lim_{N\to\infty}\frac{1}{N}K_N(x,x)$ as well as local scaling limits of the correlation kernel, leading for instance to the sine kernel in the bulk and to the Airy kernel near soft edges of the spectrum. We refer the reader to e.g.~ the seminal paper \cite{BBP} and the recent paper \cite{LiuWangZhang} where such an asymptotic analysis has been successfully carried out in detail in settings close to ours. 

We do not aim to pursue in this direction here; instead we will study different types of scaling limits that occur near hard edges of the spectrum in specific examples of biorthogonal ensembles of derivative type.

\subsection{LUE with additive P\'olya perturbation}

In this section, we consider the eigenvalue distribution of $A_1+\tau A_2$, where $A_1$ is an $N\times N$ LUE matrix with parameter $\nu\geq 0$ and $A_2$ is a Hermitian $N\times N$ random matrix whose eigenvalue distribution takes the form of a polynomial ensemble of additive derivative type \eqref{PE}, for some function $w$ which belongs to the space $\mathcal U_{c_-,c_+}^N$ for every natural number $N$. The corresponding function $W$ belongs to $\mathcal V_{c_-,c_+}^N$, and the re-scaled function $W(\tau z)$ belongs to $\mathcal V_{c_-/\tau,c_+/\tau}^N$.
As pointed out in Section \ref{section:resultsLUE+}, 
the eigenvalue distribution of $A_1+\tau A_2$ is a polynomial ensemble of additive derivative type, and it admits the 
correlation kernel 
\begin{equation}\label{kernelLUEPolya2}	
    K_N^{\rm pLUE}(x,x';\tau):=\oint_{\Sigma}\frac{\dv u}{2\pi i}\int_{c+i\mathbb R}\frac{\dv v}{2\pi i}\frac{v^N(1-u)^{N+\nu}W\left(\tau v\right)}{u^N(1-v)^{N+\nu}W\left(\tau u\right)}\frac{e^{-xv+x'u}}{v-u},
\end{equation}
where $\Sigma$ is a simple positively oriented closed contour enclosing $0$, and  $c+i\mathbb R$ with $c_-/\tau<c<\min\{c_+/\tau,1\}$ is such that it does not intersect with $\Sigma$. In addition, in view of our upcoming asymptotic analysis, we choose $c<0$ such that $c+i\mathbb R$ lies on the left of $\Sigma$. 

We rescale the correlation kernel by evaluating it at $\frac{x}{4N}$, $\frac{x'}{4N}$ and $\tau=\frac{r}{4N}$ and obtain
\begin{multline*}
    \frac{1}{4N}K_N^{\rm pLUE}\left(\frac{x}{4N},\frac{x'}{4N};\frac{r}{4N}\right)=\frac{1}{4N}\oint_{\Sigma}\frac{\dv u}{2\pi i}\int_{c+i\mathbb R}\frac{\dv v}{2\pi i}\frac{(1-u)^{\nu}}{(1-v)^{\nu}}\frac{W\left(\frac{rv}{4N}\right)}{W\left(\frac{ru}{4N}\right)}\\
    \times\quad \frac{\exp(N\log v-N\log(1-v)-\frac{x v}{4N})}{\exp(N\log u-N\log(1-u)-\frac{x' u}{4N})}\frac{1}{v-u},
\end{multline*}
with $\frac{4N}{r}c_-<c<0$, and with $\Sigma$ enclosing $0$ without intersecting $c+i\mathbb R$. Note that the exponentials are independent of the choice of branch cuts of the logarithms, but for definiteness we take $\log z$ analytic in $\mathbb C\setminus(-\infty,0]$ and $\log(1-z)$ analytic on $\mathbb C\setminus[1,+\infty)$. The integrand does depend on the branch of $(1-v)^\nu$ if $\nu$ is not integer, and we take it analytic in $\mathbb C\setminus[1,+\infty)$ and positive for $v<1$. We take $\Sigma$ such that it does not intersect with $[1,+\infty)$.
 
Setting $u=4sN+\frac{1}{2}$ and $v=4tN+\frac{1}{2}$, we obtain
\begin{multline}\label{eq:scaledbeforeDC}
    \frac{1}{4N}K_N^{\rm pLUE}\left(\frac{x}{4N},\frac{x'}{4N};\frac{r}{4N}\right)=e^{-\frac{x-x'}{8N}}\oint_{\tilde \Sigma}\frac{\dv s}{2\pi i}\int_{\tilde c+i\mathbb R}\frac{\dv t}{2\pi i}\frac{(\frac{1}{8N}-s)^\nu}{(\frac{1}{8N}-t)^\nu}\frac{W\left(rt+\frac{r}{8N}\right)}{W\left(rs+\frac{r}{8N}\right)}\\
    \times\quad \frac{\exp\left(N\log\frac{1+\frac{1}{8Nt}}{1-\frac{1}{8Nt}}-x t\right)}{\exp\left(N\log \frac{1+\frac{1}{8Ns}}{1-\frac{1}{8Ns}}-x' s\right)}\frac{1}{t-s}.
\end{multline}
Here $\tilde\Sigma=\frac{1}{4N}\Sigma-\frac{1}{8N}$ and $\tilde c=\frac{c}{4N}-\frac{1}{8N}$. By studying the domain of analyticity of the integrand, we see that the restrictions we have for $\tilde\Sigma$ and $\tilde c$ are that $\tilde\Sigma$ encloses $-\frac{1}{8N}$ but does not intersect with $\left[\frac{1}{8N},+\infty\right)$, and that $\tilde c+i\mathbb R$ lies on the left of $\tilde\Sigma$, with $\frac{c_-}{4r}-\frac{1}{8N}<\tilde c<0$. We can thus take, as illustrated in Figure \ref{figure: contourspBessel}, $\tilde\Sigma$ to be a circle independent of $N$ passing through $0$, and $\tilde c+i\mathbb R$ to be independent of $N$ and on the left of $\tilde\Sigma$ but such that $c_-<4r\tilde c$. For definiteness, let us set $\tilde c=\frac{c_-}{5r}$ and take the radius of $\tilde\Sigma$ equal to $\frac{|\tilde c|}{3}$.

We now want to take the large $N$ limit of \eqref{eq:scaledbeforeDC}. Supposing that we can apply Lebesgue's dominated convergence theorem, we obtain by using
\begin{equation*}\lim_{N\to\infty}N\log \frac{1+\frac{1}{8Nz}}{1-\frac{1}{8Nz}}=\frac{1}{4z},\qquad z\neq 0,\end{equation*}
that
 \begin{equation*}\lim_{N\to\infty}      \frac{1}{4N}K_N^{\rm pLUE}\left(\frac{x}{4N},\frac{x'}{4N};\frac{r}{4N}\right)=\int_{\tilde \Sigma}\frac{\dv s}{2\pi i}\int_{\tilde c+i\mathbb R}\frac{\dv t}{2\pi i}\frac{(-s)^\nu}{(-t)^\nu}\frac{W\left(rt\right)}{W\left(rs\right)} \frac{\exp\left(\frac{1}{4t}-x t\right)}{\exp\left(\frac{1}{4s}-x' s\right)}\frac{1}{t-s}.
 \end{equation*}
Here we choose arguments between $-\pi$ and $\pi$, such that both $(-s)^\nu$ and $(-t)^\nu$ have the same branch cut on $[0,+\infty)$. The contour $\tilde c+i\R$ does not intersect with the branch cut, while $\tilde \Sigma$ approaches the branch point $0$ from the vertical direction.

The remaining task is to justify the use of Lebesgue's dominated convergence theorem on the integral \eqref{eq:scaledbeforeDC}. For the $t$-integral, we observe that for $x\in\mathbb R$, $t\in \tilde c+i\mathbb R$, $s\in\tilde\Sigma$ and $\nu\geq 0$, we have
\begin{equation*}
    \left|\frac{1}{\left(\frac{1}{8N}-t\right)^\nu}\exp\left(-x t\right)\frac{1}{t-s}\right|\leq \left|\frac{1}{\left(\frac{1}{8N}-t\right)^\nu}\exp\left(-\tilde c x\right)\frac{1}{t-s}\right|
\leq C(x,\nu),  
\end{equation*}
for some constant $C(x,\nu)$ that may depend on $x$ and $\nu$, but not on $N$ nor on $s,t$. For the $N$-dependent term, we use the inequality
\begin{equation}
    \left|\log\frac{1+u}{1-u}\right|\le\frac{2|u|}{1-|u|},\quad |u|<1.
\end{equation}
For large $N$ we have $|\frac{1}{8Nt}|\le 1$ for $t\in\tilde c+i\mathbb R$. Therefore,
\begin{equation}
    \left|N\log\frac{1+\frac{1}{8Nt}}{1-\frac{1}{8Nt}}\right|\le N\frac{2|\frac{1}{8Nt}|}{1-|\frac{1}{8Nt}|}\le \frac{1}{2|\tilde c|}.
\end{equation}
Hence, given $r>0$, \begin{equation*}\left|\frac{1}{\left(\frac{1}{8N}-t\right)^\nu}\exp\left(N\log\frac{1+\frac{1}{8Nt}}{1-\frac{1}{8Nt}}-x t\right)\frac{1}{t-s}W\left(rt+\frac{r}{8N}\right)\right| 
\leq 
\frac{C(x,\nu)}{2|\tilde c|}\left|W\left(rt+\frac{r}{8N}\right)\right|, \end{equation*}
and by the uniform decay condition
\begin{equation*}\sup_{c_-+\epsilon<c<c_+-\epsilon}\|z^{N-1}W(z)\|_{L^\infty(c+i\mathbb R)}<\infty\end{equation*}
for $W$,
we can bound the right hand side by, for instance, the integrable function $\tilde C(x,\nu,r)/(|t|^2+1)$ with sufficiently large $\tilde C(x,\nu,r)$. This justifies the dominated convergence for the $t$-integral.

For the $s$-integral, we first observe that, for $s$ on the contour $\tilde\Sigma$, we have $|s-\frac{\tilde c}{3}|=\frac{|\tilde c|}{3}$, which simplifies to $\re s=\frac{3}{2\tilde c}|s|^2$. Then
\begin{equation}
    \left|\frac{1-8Ns}{1+8Ns}\right|^2=1+\frac{48}{|\tilde c|}\frac{N|s|^2}{|1+8Ns|^2}.
\end{equation}
Using the inequality $|1+8Ns|^2>64 N^2|s|^2$, we have
\begin{equation}
    1\le \left|\frac{1-8Ns}{1+8Ns}\right|^2\le 1+\frac{3}{4|\tilde c|N}\le \left(1+\frac{3}{4|\tilde c|N}\right)^2.
\end{equation}
This readily gives us that $\exp\left(N\log\left|\frac{1-8Ns}{1+8Ns}\right|\right)$ is uniformly bounded by $e^{\frac{3}{4|\tilde c|}}$.
Also, since  $\tilde \Sigma$ is a bounded contour where $W\left(rs+\frac{r}{8N}\right)$ is analytic and non-zero with $s\in\tilde \Sigma$, we have $1/|W\left(rs+\frac{r}{8N}\right)|$ being bounded from above. Moreover for $s\in\tilde\Sigma$, the term $\left|\frac{1}{8N}-s\right|$ is also bounded. Hence we have
\begin{equation*}\left|\left(\frac{1}{8N}-s\right)^\nu\exp\left(N\log\frac{1-\frac{1}{8Ns}}{1+\frac{1}{8Ns}}+x's\right)\frac{1}{W\left(rs+\frac{r}{8N}\right)}\right|\leq C(x',\nu,r)\end{equation*}
for some constant $C(x',\nu,r)$ independent of $s$ and of $N$.
This justifies the use of Lebesgue's dominated convergence theorem on \eqref{eq:scaledbeforeDC} and completes the proof of Theorem \ref{theorem:LUE+}.

\subsection{Muttalib-Borodin type biorthogonal ensembles of derivative type}
Here, we consider the Muttalib-Borodin type deformations of polynomial ensembles \eqref{MB0mult} of derivative type discussed in Section \ref{section:resultsMB}.
Given $c_-<0$, we let $W$ be such that it belongs to $\mathcal V_{c_-,c_+}^N$ for every $c_+>0$ and for every natural number $N$.
We recall from \eqref{kernelMBPolya} that this ensemble, in multiplicative variables $y_1,\ldots, y_N$, has the correlation kernel
\begin{equation*}K_N^{\rm MB}(y,y';\theta,\eta):=\int_{\Sigma_N}\frac{\dv u}{2\pi i}\int_{c+i\mathbb R}\frac{\dv v}{2\pi i}\frac{W\left(v\right)\prod_{j=1}^N(v-\theta j-\eta)}{W\left(u\right)\prod_{j=1}^N(u-\theta j-\eta)}\frac{y^{-v-1}(y')^{u}}{v-u},\qquad y,y'>0,\end{equation*}
where $\Sigma_N$ is closed and goes around $\theta+\eta,2\theta+\eta,\ldots, N\theta+\eta$, while it does not intersect with $c+i\mathbb R$. It will be convenient to set $c=\eta$, with $\Sigma_N$ in the half plane $\Re u>\eta$.
We also assume that there exist $c>\frac{\pi}{2\theta}$, $C,\tilde C>0$ and a semi-infinite horizontal strip $A$ containing $[\theta+\eta,+\infty)$ such that \eqref{eq:MBcondW1} and \eqref{eq:MBcondW2} hold.

Note that
\begin{equation*}
        \frac{\prod_{j=1}^N(v - \theta j-\eta )}{\prod_{j=1}^N(u - \theta j-\eta )} = \frac{\Gamma(1-\frac{u-\eta}{\theta})}{\Gamma(1-\frac{v-\eta}{\theta})}\frac{\Gamma(N +1- \frac{v-\eta}{\theta})}{\Gamma(N +1- \frac{u-\eta}{\theta})},
    \end{equation*}
    such that we can rewrite the kernel as
    \begin{equation*}K_N^{\rm MB}(y,y';\theta,\eta):=\int_{\Sigma_N}\frac{\dv u}{2\pi i}\int_{\eta+i\mathbb R}\frac{\dv v}{2\pi i}\frac{W\left(v\right)}{W\left(u\right)}\frac{\Gamma(1-\frac{u-\eta}{\theta})}{\Gamma(1-\frac{v-\eta}{\theta})}\frac{\Gamma(N + 1-\frac{v-\eta}{\theta})}{\Gamma(N + 1-\frac{u-\eta}{\theta})}\frac{y^{-v-1}(y')^{u}}{v-u}.\end{equation*}
After re-scaling,
    \begin{equation}\label{eq:kernelMBbeforeDC}\frac{1}{N^{1/\theta}}K_N^{\rm MB}\left(\frac{y}{N^{1/\theta}},\frac{y'}{N^{1/\theta}};\theta,\eta\right):=\int_{\Sigma_N}\frac{\dv u}{2\pi i}\int_{\eta+i\mathbb R}\frac{\dv v}{2\pi i}\frac{W\left(v\right)}{W\left(u\right)}\frac{\Gamma(1-\frac{u-\eta}{\theta})}{\Gamma(1-\frac{v-\eta}{\theta})}N^{\frac{v-u}{\theta}}\frac{\Gamma(N +1- \frac{v-\eta}{\theta})}{\Gamma(N + 1-\frac{u-\eta}{\theta})}\frac{y^{-v-1}(y')^{u}}{v-u}.\end{equation}
Using Stirling's approximation, it is straightforward to verify that
\begin{equation}\lim_{N\to\infty}\frac{\Gamma(N +1- \frac{v-\eta}{\theta})}{\Gamma(N +1- \frac{u-\eta}{\theta})}N^{(v-u)/\theta}=1.\end{equation}
If we are allowed to interchange the large $N$ limit with the double integral, we obtain 
   \begin{equation}\label{eq:limitMB}\lim_{N\to\infty}\frac{1}{N^{1/\theta}}K_N^{\rm MB}\left(\frac{y}{N^{1/\theta}},\frac{y'}{N^{1/\theta}};\theta,\eta\right):=\int_{\Sigma}\frac{\dv u}{2\pi i}\int_{\eta+i\mathbb R}\frac{\dv v}{2\pi i}\frac{W\left(v\right)}{W\left(u\right)}\frac{\Gamma(1-\frac{u-\eta}{\theta})}{\Gamma(1-\frac{v-\eta}{\theta})}\frac{y^{-v-1}(y')^{u}}{v-u}=\mathbb K^{\theta,\eta}(y,y'),\end{equation}
where $\Sigma$ is now unbounded and encloses all $k\theta+\eta$ for $k\in\mathbb N$.

The remaining part of the proof consists of justifying this step, which is a delicate task.
We first notice that we can evaluate the $u$-integrals in \eqref{eq:kernelMBbeforeDC} using the residue theorem. For this, we observe that the poles of the $u$-integral are $u_k=k\theta+\eta$, where $k$ runs from $1$ to $N$. They are all simple (note that $W(u_k)\neq 0$ by \eqref{eq:MBcondW2}) and we compute
\begin{equation}
    {\rm Res}\left(\Gamma\left(1-\frac{u-\eta}{\theta}\right);u_k\right)=\lim_{u\to u_k}(u-u_k)\Gamma\left(1-\frac{u-\eta}{\theta}\right)=\theta \lim_{z\to k}(z-k)\Gamma(1-z)=\frac{(-1)^{k-1}\theta}{(k-1)!}.
\end{equation}
We can thus use the residue theorem to rewrite 
the rescaled kernel \eqref{eq:kernelMBbeforeDC} as
\begin{equation}
\begin{split}
    \frac{1}{N^{1/\theta}}K\left(\frac{y}{N^{1/\theta}},\frac{y'}{N^{1/\theta}};\theta,\eta\right)
&=\sum_{k=1}^{N}\frac{(-1)^{k-1}\theta}{(k-1)!}\frac{\Gamma(N+1)(y')^{u_k}}{W(u_k)\Gamma(N+1-k)N^{\frac{u_k-\eta}{\theta}}}\\
&\quad\quad\quad\times\quad\int_{\eta+i\R}\frac{\dv v}{2\pi i}\frac{W(v)}{\Gamma(1-\frac{v-\eta}{\theta})}\frac{\Gamma(N+1-\frac{v-\eta}{\theta})}{\Gamma(N+1)}N^{\frac{v-\eta}{\theta}}\frac{y^{-v-1}}{v-u_k}.
\end{split}\label{eq:sum}
\end{equation}
We first look at the integral on the last line for fixed $k$. 
By the classical bounds
\begin{equation}
    \begin{split}
    \left|\frac{\Gamma(N+1-it)}{\Gamma(N+1)}\right|\le 1,\qquad |N^{it}|=1,\qquad\text{ for } t=\frac{v-\eta}{i\theta}\in\mathbb R,
    \end{split}
\end{equation}
the integrability of $W(v)/\Gamma\left(1-\frac{v-\eta}{\theta}\right)$ (which follows from \eqref{eq:MBcondW1} - recall Remark \ref{remark:convergence}) and the fact that $\frac{y^{-v-1}}{v-u_k}$ is uniformly bounded for $v\in\eta+i\mathbb R$, we can use the dominated convergence theorem; since \[\lim_{N\to\infty}\frac{\Gamma(N+1-\frac{v-\eta}{\theta})}{\Gamma(N+1)}N^{(v-\eta)/\theta}=1,\]
we obtain
\[\lim_{N\to\infty}\int_{\eta+i\R}\frac{\dv v}{2\pi i}\frac{W(v)}{\Gamma(1-\frac{v-\eta}{\theta})}\frac{\Gamma(N+1-\frac{v-\eta}{\theta})}{\Gamma(N+1)}N^{\frac{v-\eta}{\theta}}\frac{y^{-v-1}}{v-u_k}=\int_{\eta+i\R}\frac{\dv v}{2\pi i}\frac{W(v)}{\Gamma(1-\frac{v-\eta}{\theta})}\frac{y^{-v-1}}{v-u_k}.\]
Moreover, the left hand side is for finite $N$ bounded in modulus by 
\[y^{-\eta-1}\int_{\eta+i\R}\frac{|\dv v|}{2\pi }\frac{|W(v)|}{|\Gamma(1-\frac{v-\eta}{\theta})|}\frac{1}{|v-u_k|}\leq \frac{1}{k\theta}y^{-\eta-1}\int_{\eta+i\R}\frac{|\dv v|}{2\pi }\frac{|W(v)|}{|\Gamma(1-\frac{v-\eta}{\theta})|}.\]
It follows that the $k$-th summand in \eqref{eq:sum} is bounded in modulus by
\[\frac{y^{-\eta-1}I}{k (k-1)!}\frac{\Gamma(N+1)(y')^{u_k}}{|W(u_k)|\Gamma(N+1-k)N^{\frac{u_k-\eta}{\theta}}}= {y^{-\eta-1}I}\frac{N!}{k!(N-k)!}\frac{(y')^{k\theta+\eta}}{|W(k\theta+\eta)| N^{k}},\]
where $I=\int_{\eta+i\R}\frac{|\dv v|}{2\pi }\frac{|W(v)|}{|\Gamma(1-\frac{v-\eta}{\theta})|}$.
Also, by \eqref{eq:MBcondW2},
$\frac{1}{|W(k\theta+\eta)|}\leq \tilde Ce^{C\theta k+C\eta}$ for some $C,\tilde C>0$.
Thus, since 
\begin{align*}
{y^{-\eta-1}(y')^\eta I}\sum_{k=1}^N\frac{N!}{k!(N-k)!}\frac{(y')^{k\theta}}{|W(k\theta+\eta)| N^{k}}&\leq 
{y^{-\eta-1}(e^Cy')^\eta I}\sum_{k=1}^N\frac{N!}{k!(N-k)!}\left(\frac{(y')^{\theta}e^{C\theta}}{N}\right)^{k}\\
&={y^{-\eta-1}(e^Cy')^\eta}I\left(\left(1+\frac{e^{C\theta}}{N}\right)^{N}-1\right)\end{align*}
is uniformly bounded, we can apply Lebesgue's dominated convergence theorem to the $k$-sum of \eqref{eq:sum} to obtain
\[\lim_{N\to\infty}  \frac{1}{N^{1/\theta}}K\left(\frac{y}{N^{1/\theta}},\frac{y'}{N^{1/\theta}};\theta,\eta\right)
=\sum_{k=1}^{\infty}\frac{(-1)^{k-1}\theta}{(k-1)!}\frac{(y')^{u_k}}{W(u_k)}\int_{c+i\R}\frac{\dv v}{2\pi i}\frac{W(v)}{\Gamma(1-\frac{v-\eta}{\theta})}\frac{y^{-v-1}}{v-u_k}.\]
This is the same as 
\begin{align*}&\lim_{N\to\infty}\sum_{k=1}^{N}\frac{(-1)^{k-1}\theta}{(k-1)!}\frac{(y')^{u_k}}{W(u_k)}\int_{c+i\R}\frac{\dv v}{2\pi i}\frac{W(v)}{\Gamma(1-\frac{v-\eta}{\theta})}\frac{y^{-v-1}}{v-u_k}\\&\qquad =\lim_{N\to\infty}\int_{\Sigma_N}\frac{\dv u}{2\pi i}\int_{\eta+i\mathbb R}\frac{\dv v}{2\pi i}\frac{W\left(v\right)}{W\left(u\right)}\frac{\Gamma(1-\frac{u-\eta}{\theta})}{\Gamma(1-\frac{v-\eta}{\theta})}\frac{y^{-v-1}(y')^{u}}{v-u}\\
&\qquad =\int_{\Sigma}\frac{\dv u}{2\pi i}\int_{\eta+i\mathbb R}\frac{\dv v}{2\pi i}\frac{W\left(v\right)}{W\left(u\right)}\frac{\Gamma(1-\frac{u-\eta}{\theta})}{\Gamma(1-\frac{v-\eta}{\theta})}\frac{y^{-v-1}(y')^{u}}{v-u},
\end{align*}
where $\Sigma_N$ denotes a contour enclosing $k\theta+\eta$ for $k=1,\ldots, N$ but not enclosing $k\theta+\eta$ for $k=N+1,N+2,\ldots $, and where we used the residue theorem to reconstruct the $u$-integral. We recognize the double integral as $\mathbb K^{\theta,\eta}(y,y')$ and Theorem \ref{thm:MB} is proven.

\appendix

\section{Different expressions for the hard edge Bessel kernel}\label{appendix}

The standard expressions for the hard edge Bessel kernel, see e.g.\ \cite{Forrester_book}, are  
\begin{equation}\label{Besselkernel1}
K_\nu(x,x')
=
\frac{
\sqrt{x'}\, J_\nu(\sqrt{x}) J_\nu'(\sqrt{x'})
-
\sqrt{x}\, J_\nu(\sqrt{x'}) J_\nu'(\sqrt{x})
}{
2(x-x')
},\qquad \nu>-1,\quad x,x'>0,
\end{equation}
and
\begin{equation}\label{Besselkernel2}
K_\nu(x,x')=\frac{1}{4}\int_0^1 J_\nu(\sqrt{xr})J_\nu(\sqrt{x'r})\dv r,\qquad \nu>-1,\quad x,x'>0.
\end{equation}
Here $J_\nu$ denotes the Bessel function of the first kind, which can be expressed as a contour integral known as Sch\"afli's integral~\cite[Formula (10.9.19)]{DLMF},
\begin{equation}\label{Schafli}
J_\nu(\sqrt{z})
=
\frac{1}{2\pi i}\left(\frac{z}{4}\right)^{\nu/2}
\int_{\gamma}
v^{-\nu-1}
e^{v-\frac{z}{4v}}
\, \dv v,\qquad \nu>-1,\ z>0,
\end{equation}
where $\gamma$ encloses the branch cut $(-\infty,0]$ in the counterclockwise direction. Note that for integer $\nu$, the integrand does not have a branch cut, such that one may equivalently take $\gamma$ to be a closed loop around $0$. 
We will see that our expression for the Bessel kernel~\eqref{eq:Besselkernel} is equivalent to \eqref{Besselkernel2} for $\nu\geq 0$.

Setting $u=\frac{z}{4v}$ in \eqref{Schafli}, we obtain 
\begin{equation}\label{Schafli2}
J_\nu(\sqrt{z})
=-
\frac{1}{2\pi i}\left(\frac{z}{4}\right)^{-\nu/2}
\oint_{\Sigma}
u^{\nu-1}
e^{-u+\frac{z}{4u}}\dv u,
\end{equation}
where $\Sigma$ is the image of $\gamma$ under the map $v\mapsto \frac{z}{4v}$ (for $z>0$). In other words, it is a closed loop going through $0$ and not intersecting $(-\infty,0)$, oriented clockwise and smooth except at $0$. This integral is absolutely integrable for $\nu>-1$ if $\Sigma$ approaches the essential singularity $0$ with angles $\pi\pm \theta$, $0\leq \theta<\pi$.

Substituting \eqref{Schafli} and \eqref{Schafli2} into \eqref{Besselkernel2} and changing variables $v=-xt$, $u=-x's$ (with $x,x'>0$), we obtain 
\begin{equation}\label{triple}
K_\nu(x,x')
=-\frac{1}{4}\left(\frac{x'}{x}\right)^{\nu/2}
\int_0^1\dv r \left(\frac{1}{2\pi i}
\int_{\tilde\gamma}
(-t)^{-\nu-1}
e^{-xt+\frac{r}{4t}}\dv t\right)\left(\frac{1}{2\pi i}
\oint_{\tilde \Sigma}
(-s)^{\nu-1}
e^{x's-\frac{r}{4s}}\dv s\right),
\end{equation}
where $\tilde\gamma$ encloses $[0,+\infty)$ with clockwise orientation, and $\tilde\Sigma$ is a closed positively oriented loop going through $0$ and not intersecting $(0,+\infty)$, smooth except at $0$ and approaching $0$ with angles $\pm \theta$, $0\leq \theta<\pi$ (note that we have changed the orientation of both contours).
We moreover choose $\tilde\gamma$ and $\tilde\Sigma$ such that they do not intersect. By Fubini's theorem, we can bring the $r$-integral inside and compute it explicitly. We then obtain 
\begin{equation*}
K_\nu(x,x')
=\left(\frac{x'}{x}\right)^{\nu/2}
\int_{\tilde\gamma}\frac{\dv t }{2\pi i}\oint_{\tilde \Sigma}\frac{\dv s}{2\pi i}
\frac{(-s)^\nu}{(-t)^\nu}\frac{e^{-xt}}{e^{-x's}}\frac{e^{\frac{1}{4t}-\frac{1}{4s}}}{t-s}
-\left(\frac{x'}{x}\right)^{\nu/2}
\int_{\tilde\gamma}\frac{\dv t }{2\pi i}\oint_{\tilde \Sigma}\frac{\dv s}{2\pi i}
\frac{(-s)^\nu}{(-t)^\nu}\frac{e^{-xt}}{e^{-x's}}\frac{1}{t-s}.
\end{equation*}
The second term in the above expression vanishes by Cauchy's theorem because the $s$-integrand is analytic inside $\tilde \Sigma$ and continuous up to the boundary. The first term does not vanish because the $s$-integrand is unbounded inside $\tilde\Sigma$ near the essential singularity at $0$. We finally obtain 
\begin{equation}
K_\nu(x,x')
=\left(\frac{x'}{x}\right)^{\nu/2}
\int_{\tilde\gamma}\frac{\dv t }{2\pi i}\oint_{\tilde \Sigma}\frac{\dv s}{2\pi i}
\frac{(-s)^\nu}{(-t)^\nu}\frac{e^{-xt+\frac{1}{4t}}}{e^{-x's+\frac{1}{4s}}}\frac{1}{t-s}.
\end{equation}
We also observe that the $t$-integrand is $\mathcal O(|t|^{-\nu-1})$ as $t\to\infty$ in the half plane $\Re z\geq c$ for $c>0$, such that we can deform the contour $\tilde\gamma$ to the vertical line $c+i\mathbb R$. Similarly, the $s$-integrand is $\mathcal O(|s|^{\nu})$ as $s\to 0$ with $\Re s> 0$, such that we can deform $\tilde\Sigma$ to a curve passing vertically through $0$.
Up to the cocycle factor $\left(\frac{x'}{x}\right)^{\nu/2}$, this yields precisely our expression \eqref{eq:Besselkernel}.


\section*{Acknowledgments}
The authors are grateful to Peter Forrester and Mario Kieburg for fruitful discussions. TC
acknowledges support by FNRS Research Project T.0028.23 and by the Fonds Sp\'ecial de Recherche
of UCLouvain. JZ acknowledges partial financial support of FWO Odysseus grant No. G0DDD23N and FWO Fellowship Junior 1234325N.

\end{document}